\documentclass[useAMS,referee,usenatbib]{biom}

\usepackage{amsmath} 
\usepackage{amssymb} 
\usepackage{mathrsfs} 
\usepackage{xcolor} 
\usepackage{multirow,booktabs} 
\usepackage{graphicx} 
\usepackage{hyperref} 

\def\bx{\mathbf{x}}
\def\bX{\mathbf{X}}

\def\bW{\mathbf{W}}

\def\bZ{\mathbf{Z}}
\def\bz{\mathbf{z}}
\def\bv{\mathbf{v}}
\def\bV{\mathbf{V}}
\def\bu{\mathbf{u}}
\def\bU{\mathbf{U}}
\def\bh{\mathbf{h}}
\def\bshat{\widehat{\mathbf{s}}}
\def\bShat{\widehat{\mathbf{S}}}
\def\bSbar{\bar{\mathbf{S}}}

\def\Ydag{Y^{\dagger}}
\def\Ytld{\widetilde{Y}}
\def\Yhatdag{\widehat{Y}^{\dagger}}
\def\bvphi{\bmath{\varphi}}

\def\bgam{\bmath{\gamma}}

\def\bgamhat{\widehat{\bmath{\gamma}}}
\def\bgambar{\bar{\bmath{\gamma}}}
\def\balph{\bmath{\alpha}}
\def\balphhat{\widehat{\bmath{\alpha}}}
\def\balphbar{\bar{\bmath{\alpha}}}

\def\bbeta{\bmath{\beta}}
\def\bbetahat{\widehat{\bmath{\beta}}}
\def\bbetabar{\bar{\bmath{\beta}}}
\def\bGam{\bmath{\Gamma}}
\def\beps{\bmath{\epsilon}}
\def\Delthat{\widehat{\Delta}}
\def\Deltbar{\bar{\Delta}}

\def\Deltabarstr{\bar{\Delta}^*}
\def\muhat{\widehat{\mu}}
\def\mubar{\bar{\mu}}

\def\thetastr{\theta^*}
\def\mubarstr{\bar{\mu}^*}
\def\bzero{\bmath 0}
\def\bone{\bmath 1}
\def\bI{\mathbf{I}}
\def\pihat{\widehat{\pi}}

\def\Ystrhat{\widehat{Y}^*}
\def\Wscrtld{\widetilde{\mathcal{W}}}

\def\Wscrhat{\widehat{\mathcal{W}}}

\def\Gscr{\mathcal{G}}
\def\Lscr{\mathscr{L}}
\def\Uscr{\mathscr{U}}
\def\Dscr{\mathscr{D}}

\def\Sscr{\mathcal{S}}
\def\Mscr{\mathcal{M}}
\def\Pscr{\mathcal{P}}

\def\Lscralt{\mathcal{L}}
\def\ntot{N}
\def\ntotinv{N^{-1}}

\def\nhlfNinv{\frac{\sqrt{n}}{N}}
\def\trans{^{\sf \tiny T}}
\def\indep{\perp\!\!\!\perp}
\def\ddpi{\frac{\partial}{\partial\pi}}
\def\ddu{\frac{\partial}{\partial u}}
\def\ddbgam{\frac{\partial}{\partial\bgam}}
\def\ddbgamT{\frac{\partial}{\partial\bgam\trans}}
\def\ddbalph{\frac{\partial}{\partial\balph\trans_1}}
\def\ddbbeta{\frac{\partial}{\partial\bbeta\trans_1}}
\def\ddtheta{\frac{\partial}{\partial\theta}}
\def\dduK{\dot{K}}
\def\atld{\tilde{a}}
\def\btld{\tilde{b}}
\def\fstr{f^*}

\def\bvthetbar{\bar{\bmath{\vartheta}}}
\def\bvthethat{\widehat{\bmath{\vartheta}}}

\def\Upi{U_{\pi}}

\def\bZpi{\bZ_{\pi}}
\def\bZpii{\bZ_{\pi,i}}
\def\bZpihati{\bZ_{\pihat,i}}
\def\bZpihatstri{\bZ_{\pihat^*,i}}
\def\gdot{\dot{g}}
\def\bzpi{\bz_{\pi}}

\def\SSDR{\mbox{SS}_{\sf \scriptscriptstyle DR}}

\def\CCDR{\mbox{CC}_{\sf \scriptscriptstyle AIPW}}
\def\CCNaive{\mbox{CC}_{\sf \scriptscriptstyle Naive}}

\def\SSPrePost{\mbox{SS}_{\sf \scriptscriptstyle AIPW}}

\def\E{\mathbb{E}}
\def\P{\mathbb{P}}

\def\omegahat{\widehat{\omega}}
\def\omegabar{\bar{\omega}}
\def\half{1/2}

\def\Upi{U_{\pi}}

\def\bZpi{\bZ_{\pi}}
\def\bZpii{\bZ_{\pi,i}}
\def\bZpihati{\bZ_{\pihat,i}}
\def\bZpihatstri{\bZ_{\pihat^*,i}}
\def\gdot{\dot{g}}
\def\bzpi{\bz_{\pi}}

\def\DelthatAIPW{\Delthat_{\sf \scriptscriptstyle AIPW}}
\def\muhatoneAIPW{\muhat_{\sf \scriptscriptstyle 1,AIPW}}
\def\muhatzerAIPW{\muhat_{\sf \scriptscriptstyle 0,AIPW}}
\def\muhatkAIPW{\muhat_{\sf \scriptscriptstyle k,AIPW}}

\def\sqrtn{n^{1/2}}
\def\Tscr{\mathcal{T}}

\definecolor{darkred}{RGB}{150,50,50}

\def\bSig\mathbf{\Sigma}

\newcommand{\norm}[1]{\left\lVert#1\right\rVert}
\newcommand{\abs}[1]{\left|#1\right|}

\DeclareMathOperator*{\argmin}{argmin}

\newenvironment{eq} 
{
\align
}
{
\endalign
}
\newenvironment{eq*} 
{
\csname align*\endcsname
}
{
\csname endalign*\endcsname
}

\title[Semi-Supervised Estimation of Average Treatment Effects]{Robust and Efficient Semi-Supervised Estimation of Average Treatment Effects with Application to Electronic Health Records Data}

\author{David Cheng$^{1}$, 
Ashwin N. Ananthakrishnan$^{2}$, and Tianxi Cai$^{3,*}$\email{tcai@hsph.harvard.edu} \\
$^{1}$VA Boston Healthcare System, Boston, Massachusetts, U.S.A.\\
$^{2}$Division of Gastroenterology, Massachusetts General Hospital, Boston, Massachusetts, U.S.A.\\
$^{3}$Department of Biostatistics, Harvard T.H. Chan School of Public Health, Boston, Massachusetts, U.S.A.}

\begin{document}

\date{}

\label{firstpage}

\begin{abstract}
We consider the problem of estimating the average treatment effect (ATE) in a semi-supervised learning setting, where a 
very small proportion of the entire set of observations are labeled with the true outcome but features predictive of the outcome are available
among all observations. This problem arises, for example, when estimating treatment effects in electronic health records 
(EHR) data because gold-standard outcomes are often not directly observable from the records but are observed for a limited
number of patients through small-scale manual chart review. We develop an imputation-based
approach for estimating the ATE that is robust to misspecification of the imputation model. This effectively allows information
from the predictive features to be safely leveraged to improve efficiency in estimating the ATE.  The estimator is additionally
doubly-robust in that it is consistent under correct specification of either an initial propensity score model or a
baseline outcome model.  It is also locally semiparametric efficient under an ideal semi-supervised model where the distribution
of the unlabeled data is known. Simulations exhibit the efficiency and robustness of the proposed method compared to existing 
approaches in finite samples.We illustrate the method by comparing rates of treatment response to two biologic agents for 
treatment inflammatory bowel disease using EHR data from Partner's Healthcare.
\end{abstract}

\begin{keywords}
causal inference, double-robustness, missing data, semiparametric efficiency, semi-supervised learning, surrogate outcomes.
\end{keywords}

\maketitle

\section{Introduction}
\label{s:intro}

There is often interest in estimating the average treatment effect (ATE) of a binary treatment $T$ on 
an outcome $Y$ that is observed among a very limited subset of observations but can be approximated by surrogate variables $\bW$ 
available among all observations.
As a motivating example, we consider comparing the outcomes of two treatments in electronic health record (EHR) data,
where $Y$ can be a clinical outcome of interest not directly encoded in patients' medical records and $\bW$ are post-treatment 
features that can be automatically extracted from the charts, such as the receipt of billing or procedure codes and mentions
of selected terms from physicians' notes.  Though different study designs are possible, we assume for phenotyping purposes $Y$ is collected for a small \emph{random} subset of all patients, which constitute the labeled data $\Lscr$,
through manual chart review. It may not be possible to comprehensively label the data because chart review is a costly and
time-consuming process.

A common strategy for analyzing such data is to use the surrogates $\bW$ to approximate $Y$ by some imputed outcome 
$\Ydag = g(\bW)$.  The imputation can be based on heuristic rules determined by domain-specific knowledge 
(e.g. presence of certain set of diagnostic codes) or an imputation model that predicts $Y$ given $\bW$ trained using the labeled data $\Lscr$, which includes observations of both $\bW$ and $Y$ \citep{ananthakrishnan2016comparative}.  
However, for complex outcomes it may be difficult to obtain an accurate imputation $\Ydag$ because $\bW$ have limited 
predictive power or the functional form of the imputation model may be difficult to correctly specify.  
When the imputation quality is inadequate, it is often unclear whether using the inaccurate imputations $\Ydag$ in subseqent analyses can lead to biased 
estimates of the ATE on the \emph{actual} outcome $Y$.

Previously, related methods have been developed in the surrogate outcomes literature to leverage both the labeled
data $\Lscr$ and the unlabeled data $\Uscr$, which includes all variables in $\Lscr$ except for $Y$, for estimating 
regression parameters 
\citep{pepe1992inference} 
and solutions 
to estimating equations \citep{chen2003information}. But these methods tend to assume a univariate 
surrogate with low dimensional baseline covariates $\bX$. 
Alternatively the problem can be viewed as estimating the mean of a longitudinal outcome subject to monotone missingness, where
$\bW$ is an outcome at an initial time point and $Y$ is the final outcome. In this context semiparametric efficiency 
theory has been developed to identify efficient estimators under various semiparametric models \citep{robins1994estimation,rotnitzky1998semiparametric},
which has lead to development of doubly-robust (DR) augmented IPW (AIPW) estimators  in different problems 
\citep{davidian2005semiparametric,williamson2012doubly,zhang2016causal}. In particular, \cite{davidian2005semiparametric}
develops such an efficient estimator for estimating the effect of a randomized treatment where the final outcome is subject to missingness
but intermediate outcomes are collected for all patients. With some minor modification this estimator could also leverage 
the unlabled data $\Uscr$ to aid estimation of the ATE, and we consider it as a reference method in the simulations.  
However, this method was developed under a data model commonly employed in missing data problems with independent and identically 
(iid) observations and a probability of missingness that is bounded away from $0$.
The data model we consider here differs in that:
(1) we assume that the number of labeled observations $n$ is fixed by design, and (2) 
we make a \emph{semi-supervised} assumption that the proportion of labeled data $\nu_n$ tends to $0$ as $n\to\infty$, to 
reflect the large size of $\Uscr$ relative to that of $\Lscr$. These features complicate conventional applications of
semiparametric efficiency theory, and the efficiency and finite sample performance of existing estimators may not be clear.

In this paper, we propose a semi-supervised (SS) estimator for the ATE 
based on an imputation followed by inverse probability weighting (IPW).  It is doubly-robust and locally semiparametric efficient under an ideal model approximating the data distribution of the semi-supervised setup.
The imputations are constructed such that the resulting estimator is robust to misspecification of the imputation model, enabling $\Uscr$ to be safely used to improve the estimation.  We further employ a double-index propensity score \citep{cheng2017estimating} for additional robustness and possible small-sample efficiency gains.  The remainder of the paper is organized as follows.  We formalize the SS estimation problem in Sections 2.1-2.2 and develop the estimator in Sections 2.3-2.5.  A perturbation resampling procedure is proposed in Section 2.6 for inference. Section \ref{s:sims} presents simulations showing the robustness and efficiency of the proposed estimator, and Section \ref{s:data} applies the method to compare two biologic therapies for treating inflammatory bowel disease (IBD) in EMR data from Partner's Healthcare.  We conclude with some remarks in Section 5.  Proofs are deferred to the Web Appendices.

\section{Method}
\label{s:method}

\subsection{Notations and Semi-Supervised Framework}
Let $Y$ denote an outcome,
$T\in \{ 0,1\}$ a binary treatment, $\bX$ a $p_x$-dimensional vector of pre-treatment baseline covariates, $\bW$ a $p_w$-dimensional vector of post-treatment surrogate variables that are potentially predictive of $Y$, and $\bV = (\bW\trans,\bX\trans)\trans$.  
For example, in the EHR context, $\bX$ may include demographics and prior comorbidities that may confound 
naive associations between $T$ and $Y$, while $\bW$ may be counts of post-treatment codes or terms.
The labeled data consists of $n$ iid observations $\Lscr = \{ (Y_i,T_i,\bV_i\trans)\trans:i=1,\ldots,n\}$, while the unlabeled data consists of $N-n$ iid observations without $Y$, $\Uscr = \{ (T_i,\bV_i\trans)\trans : i=n+1,\ldots,\ntot\}$, with $\Uscr \indep \Lscr$ and $n$ and $N$ fixed.   
We assume that $n$ observations were randomly selected for labeling so that $Y$ is essentially missing completely at random (MCAR) from observations in $\Uscr$.  In the SS setting $N \gg n$ so that $\nu_n = n/N \to 0$ as $n\to \infty$.
The entire observed data $\Dscr = \Lscr \cup \Uscr$ could thus be framed as $\Dscr = \{ (L_i,\Ytld_i,T_i,\bV_i\trans): i=1,\ldots,N\}$, where $L$ is an indicator of labeling such that $\Ytld = Y$ when $L=1$ and $\Ytld$ is an arbitrary value otherwise with $L \indep (\Ytld, T, \bV)$. But unlike traditional missing data frameworks, $L$ is constrained such that $\sum_{i=1}^N L_i = n$,
where $n$ and $N$ satisfy $\nu_n = n/N \to 0$.

\subsection{Target Parameter and Leveraging Unlabeled Data}
Let $Y^{(1)}$ and $Y^{(0)}$ denote the counterfactual outcomes had an individual received treatment or control.  Based on the observed data $\Dscr$ we want to estimate the ATE:
\begin{eq}
\Delta = \E\{Y^{(1)}\} - \E\{Y^{(0)}\} = \mu_1 - \mu_0.
\end{eq}
We require the following standard assumptions to identify $\Delta$:
\begin{eq}
&Y = TY^{(1)} +(1-T)Y^{(0)} \label{e:assump}\\
&(Y^{(1)},Y^{(0)})\indep T \mid \bX \label{e:assumpb}\\ 
& \pi(\bx) \in [\epsilon_{\pi}, 1-\epsilon_{\pi}] \text{ for some $\epsilon_{\pi}>0$ when } f(\bx)>0, \label{e:assumpc}
\end{eq}
where $\pi(\bx)= \P(T=1\mid\bX=\bx) $ is the PS and $f(\bx)$ is the joint density for the covariates.  In the typical setting where the outcome is fully observed, the ATE can be identified through the g-formula \citep{robins1986new} for a point exposure:
\begin{eq}
\label{e:gform}
\Delta =\E\{ \mu_1(\bX) - \mu_0(\bX)\}=\E\left\{ \frac{I(T=1)Y}{\pi(\bX)}-\frac{I(T=0)Y}{1-\pi(\bX)}\right\},
\end{eq}
where $\mu_k(\bx)=\E(Y\mid\bX=\bx,T=k)$ for $k=0,1$.  This suggests the usual estimators based on averaging the outcome weighted by IPW weights or averaging estimated outcome models.  When the outcome is missing but surrogates $\bW$ are observed, the more general g-formula for longitudinal studies can be applied to show that:
\begin{eq*}
\Delta &= \E\left[ \E\{\xi_1(\bV)\mid\bX,T=1\}-\E\{\xi_0(\bV)\mid\bX,T=0\}\right] \\
&= \E\left\{ \frac{I(T=1)\xi_1(\bV)}{\pi(\bX)} - \frac{I(T=0)\xi_0(\bV)}{1-\pi(\bX)}\right\},
\end{eq*}
where $\xi_k(\bv)=\E(\Ytld \mid \bV=\bv, T=k, L=1)=\E(Y\mid \bV=\bv,T=k)$ for $k=0,1$. This decomposition suggests that, if a consistent estimator for $\xi_k(\bv)$ is available, then $\Delta$ can be estimated by first imputing $Y$ through the $\xi_k(\bv)$ estimator and then applying IPW or outcome regression methods to the imputed outcome. 
However, obtaining a consistent estimator for $\xi_k(\bv)$ may not be feasible without strong modeling assumptions due to the potential high dimensionality of $\bv$ and complexity of the functional form of $\xi_k(\bv)$. In the following we show that even with incorrectly specified models for $\xi_k(\bv)$, it is still possible to leverage $\Uscr$ in estimating $\Delta$ without introducing bias from their misspecification.

\subsection{Robust Imputations}
Let $\Upi = I(T=1)/\pi(\bX)-I(T=0)/\{ 1-\pi(\bX)\}$ denote a utility covariate given $\pi(\bx)$, assumed momentarily to be known.  We now argue that inclusion of such a covariate in an imputation model yields imputations that are robust to misspecification of the model.  Suppose we postulate a parametric \emph{working model}, possibly misspecified, for $\xi_k(\bv)$:
\begin{eq}
\label{e:impmod}
\xi_T(\bV) = g_{\xi}(\gamma_0+\bgam_1\trans \bh(\bV)+\gamma_2T+\gamma_3\Upi) = g_{\xi}(\bgam\trans\bZpi),
\end{eq}
where $\bgam=(\gamma_0,\bgam_1\trans,\gamma_2,\gamma_3)\trans$, $\bZpi = (1,\bV\trans,T,\Upi)\trans$, $g_{\xi}(\cdot)$ is a specified link function, and $\bh(\cdot)$ is a vector of fixed basis expansion functions that can incorporate nonlinear effects. 
Interactions between $\bh(\bV)$ and $T$ could also be included in the specification without difficulty.
We estimate $\bgam$ as $\bgamhat$, the solution to a penalized estimating equation with ridge regularization:
\begin{eq}
\label{e:gamest}
n^{-1}\sum_{i=1}^n \bZpii \left\{ Y_i - g_{\xi}(\bgam\trans\bZpii)\right\} - \lambda_n \bgam_{\circ} = \bzero,
\end{eq}
where $\bgam_{\circ} = (0,\bgam_{\{-1\}}\trans)^{\trans}$ with $\bgam_{\{-1\}}$ is the subvector of $\bgam$ that excludes the intercept $\gamma_0$ and $\lambda_n = o(n^{-1/2})$ is a tuning parameter, which allows $\bgamhat$ to have a $n^{-1/2}$ convergence rate.
In particular, this class of estimators includes ridge estimators for GLMs based on exponential families with canonical link functions. Ridge regularization is suggested here to improve finite sample performance but is not essential for asymptotic properties discussed below. Other regularization penalties besides the ridge penalty or no regularization can also be used, as long as $\bgamhat$ maintains a $n^{-1/2}$ convergence rate. Using that $Y$ is MCAR, standard arguments  (Web Appendix A) can be used to show that $\bgamhat\overset{p}{\to}\bgambar$ where $\bgambar$ solves:
\begin{eq*}
\E\left[ \bZ_{\pi} \left\{Y- g_{\xi}(\bgam\trans\bZpi)\right\}\right]=\bzero,
\end{eq*}
with the expectation being taken \emph{over the entire population} and not restricted only to the $L=1$ subpopulation.  In particular, for $\Ydag = g_{\xi}(\bgambar\trans\bZpi)$,
since $\bZ_{\pi}$ includes $U_{\pi}$, this implies:
\begin{eq}
\label{e:robimp}
\E\left\{ \frac{I(T=1)Y}{\pi(\bX)} - \frac{I(T=0)Y}{1-\pi(\bX)}\right\}=\E\left\{ \frac{I(T=1)\Ydag}{\pi(\bX)} - \frac{I(T=0)\Ydag}{1-\pi(\bX)}\right\}.
\end{eq}
This suggests that a standard IPW estimator based on the imputed outcome $\Ydag$ 
has the \emph{same asymptotic limit} as if the true outcomes were used, even if imputation model \eqref{e:impmod} is misspecified.  Consequently the surrogate data from $\Uscr$ could be safely used to impute the outcome using estimates of these imputed outcomes $\Ydag$.

Augmentation covariates similar to $U_{\pi}$ have been previously used, for example, to construct locally semiparametric efficient doubly-robust estimators 
\citep{van2006targeted,bang2005doubly}, construct improved locally efficient doubly-robust estimators 
with enhanced efficiency under misspecification \citep{rotnitzky2012improved}, and ensure valid doubly-robust inference \citep{benkeser2017doubly}.
Our work follows these previous works in leveraging augmentation covariates to achieve desired statistical properties.
In our case, $U_{\pi}$ is used to ensure robustness of the final estimator to misspecification of the working imputation model.

In practice $\pi(\bx)$ also needs to be estimated, which is typically done through parametric modeling such as logistic regression.  When $\pi(\bx)$ is estimated by an estimator $\pihat(\bx)$, the IPW estimator discussed above will be consistent for $\Delta$ if $\pihat(\bx)$ is consistent for $\pi(\bx)$ but otherwise could be biased if the parametric model for $\pi(\bx)$ is misspecified.  
Similar arguments can be used to construct alternative imputations $\Ydag$ that could be substituted for $Y$ in an outcome regression estimator and maintain robustness to misspecification of the imputation model.  However, such an approach would then require correct specification of the outcome regression model $\mu_k(\bx)$ to be consistent for $\Delta$.  

\subsection{Doubly-Robust IPW Based on the Double-Index PS}
We adopt an IPW estimator for the final estimator but weighting with a double-index PS (DiPS), which is an alternative
method to estimating the PS $\pi(\bx)$ that calibrates an initial parametric PS estimate to allow prognostic covariates in $\bX$
to inform the PS estimation \citep{cheng2017estimating}.  
When $Y$ is fully observed, this was previously 
shown to yield a doubly-robust (DR) IPW estimator, which is consistent for $\Delta$ if a working model for either
$\pi(\bx)$ or $\mu_k(\bx)$ is correctly specified.  Adoption of DiPS in the proposed SS estimator also 
yields a DR estimator, obviating the need for alternative sets of robust imputations $\Ydag$ to accommodate estimation based on correctly
specifying models for $\pi(\bx)$ or $\mu_k(\bx)$.
Plugging in the DiPS is a natural approach to achieve double-robustness in this context,
as the rationale for robust imputations from above assumes an IPW estimator is used for the final estimate following imputation.  
Although augmented IPW estimators (AIPW) \citep{robins1994estimation} can also be used to achieve double-robustness, it
is not immediately clear whether the resulting estimator would be robust to misspecification of imputation model.

In the following we present the details in estimating DiPS and then define the final IPW estimator.  We 
postulate the following working parametric models for $\pi(\bx)$ and $\mu_k(\bx)$:
\begin{eq}
&\pi(\bX) = g_{\pi}(\alpha_0 + \balph_1\trans\bX) = \pi(\bX;\balph) \label{e:pswrk}\\
&\mu_T(\bX) = g_{\mu}(\beta_0 +\bbeta_1\trans\bX + \beta_2 T) = \mu_T(\bX;\bbeta), \label{e:orwrk}
\end{eq}
where $\balph = (\alpha_0,\balph_1\trans)\trans$, $\bbeta = (\beta_0,\bbeta_1\trans,\beta_2)\trans$, and $g_{\pi}(\cdot)$ and $g_{\mu}(\cdot)$ are specified link functions.  
Interactions between $\bX$ and $T$ in the baseline outcome model $\mu_k(\bx;\bbeta)$ can also be accommodated by
estimating the DiPS separately by treatment groups \citep{cheng2017estimating}.
To allow for a large set of covariates, $\balph$ and $\bbeta$ are estimated using regularized maximum likelihood estimators, 
for example, as in 
$\balphhat=\argmin_{\balph} \left\{ -\ntotinv\sum_{i=1}^{\ntot} \ell_{\pi}(\balph;\bX_i,T_i) + p_{\lambda_{\ntot}}(\balph_1)\right\}$ 
and $\bbetahat=\argmin_{\bbeta} \left\{ -n^{-1}\sum_{i=1}^n \ell_{\mu}(\bbeta;Y_i,\bX_i,T_i)+p_{\lambda_n}(\bbeta_{\{-1\}})\right\}$,
where $\ell_{\pi}(\balph;\bX_i,T_i)$ and $\ell_{\mu}(\bbeta;Y_i,\bX_i,T_i)$ are the log-likelihood contributions for the $i$-th 
observation, and $p_{\lambda_{\ntot}}(\cdot)$ and $p_{\lambda_n}(\cdot)$ are penalty functions chosen such that the oracle 
properties \citep{fan2001variable} hold, such as the adaptive LASSO (ALASSO) \citep{zou2006adaptive}.
We then calibrate the initial PS estimate $\pi(\bx;\balphhat)$ by the kernel smoothing estimator:
\begin{eq*}
\pihat(\bx;\balphhat_1,\bbetahat_1) =\frac{\ntotinv\sum_{j=1}^{\ntot}K_h(\bShat_j-\bshat)I(T_j=1)}{\ntotinv\sum_{j=1}^{\ntot}K_h(\bShat_j-\bshat)},
\end{eq*}
where $\bShat_j= (\balphhat_1,\bbetahat_1)\trans \bX_j$ and $\bshat=(\balphhat_1,\bbetahat_1)\trans \bx$ are bivariate scores that represent the covariate in the directions of $\balphhat_1$ and $\bbetahat_1$, $K_h(\cdot)=h^{-2}K(\cdot/h)$, with $K(\cdot)$ being a bivariate $q$-th order kernel with $q>2$ and $h=O(N^{-\alpha})$ being a bandwidth for which a suitable choice of $\alpha$ is discussed below. 
The intution is that smoothing $T$ over $\balphhat_1\trans\bX$ and $\bbetahat_1\trans\bX$ calibrates $\pi(\bx;\balphhat)$
closer to the true $\pi(\bx)$, incorporating variation in covariates in $\bX$ that are associated with $Y$ but may
have been selected out in the initial PS model due to weak association with $T$.

Finally, we define the proposed $\SSDR$ estimator as $\Delthat = \muhat_1 - \muhat_0$ where:
\begin{eq}
&\muhat_1 = \left\{ \sum_{i=1}^{\ntot}\frac{I(T_i=1)}{\pihat(\bX_i;\balphhat_1,\bbetahat_1)}\right\}^{-1}\left\{ \sum_{i=1}^{\ntot}\frac{I(T_i=1)\Yhatdag_i}{\pihat(\bX_i;\balphhat_1,\bbetahat_1)}\right\} \\
\text{ and } &\muhat_0 = \left\{ \sum_{i=1}^{\ntot}\frac{I(T_i=0)}{1-\pihat(\bX_i;\balphhat_1,\bbetahat_1)}\right\}^{-1}\left\{ \sum_{i=1}^{\ntot}\frac{I(T_i=0)\Yhatdag_i}{1-\pihat(\bX_i;\balphhat_1,\bbetahat_1)}\right\},
\end{eq}
with $\Yhatdag_i = g_{\xi}(\bgamhat\trans\bZpihati)$ 
.  The final estimator $\Delthat$ substitutes the robust imputations $\Yhatdag$, based on the PS estimated by the DiPS, 
into a standardized IPW estimator weighted also with the DiPS.  We next consider the large-sample properties of $\Delthat$, 
starting with its DR property and its influence function expansion.

\subsection{Consistency and Asymptotic Linearity of $\Delthat$}
We show in Web Appendix B that, under the causal identification assumptions \eqref{e:assump}-\eqref{e:assumpc} and mild regularity conditions,
given that  $h=O(N^{-\alpha})$ with $\alpha \in (\frac{1-\beta}{2q} , \frac{\beta}{2} \wedge \frac{1}{4})$ and $n=O(N^{1-\beta})$ with $\beta \in (\frac{1}{q+1}, 1)$,
$\Delthat$ is doubly-robust so that:
\begin{eq}
\Delthat - \Delta = O_p(n^{-1/2})
\end{eq}
when either the PS model $\pi(\bx;\balph)$ in \eqref{e:pswrk} or the baseline outcome model $\mu_k(\bx;\bbeta)$ in \eqref{e:orwrk}
is correctly specified. 
As this result does not depend on the specification of the model for $\xi_k(\bv)$, this result also verifies that $\Delthat$ is robust
to misspecification of the imputation model.
To characterize the large sample variability of $\Delthat$, we next show it is asymptotically linear and identify its influence function.  First define $\Deltbar = \mubar_1 - \mubar_0$, where:
\begin{eq*}
\mubar_1 = \E\left\{ \frac{I(T=1)\Ydag}{\pi(\bX;\balphbar_1,\bbetabar_1)}\right\} \text{ and } \mubar_0 = \E\left\{ \frac{I(T=0)\Ydag}{1-\pi(\bX;\balphbar_1,\bbetabar_1)}\right\},
\end{eq*}
with $\pi(\bx;\balphbar_1,\bbetabar_1)=\P(T=1\mid \balphbar_1\trans\bX=\balphbar_1\trans\bx,\bbetabar_1\trans\bX=\bbetabar_1\trans\bx)$, $\balphbar_1$ and $\bbetabar_1$ as the probability limits of $\balphhat_1$ and $\bbetahat_1$ regardless of model adequacy, and $\Ydag$ being defined as in \eqref{e:robimp} except that $\pi(\bx)$ is replaced by $\pi(\bx;\balphbar_1,\bbetabar_1)$. We show in Web Appendix B that, under the same requirements for $\alpha$ and $\beta$, 
the influence function for $\Delthat$ is given by the summand of $n^{1/2}(\Delthat_k - \Deltbar_k) = \Wscrhat_1 - \Wscrhat_0$, where $\Wscrhat_k = n^{1/2}(\muhat_k - \mubar_k)$ for $k=0,1$ and:
\begin{eq}
\label{e:IF}
\Wscrhat_k &= n^{-1/2}\sum_{i=1}^n (\bv_{\bbeta_1,k}\trans+\bu_{pa,\pi,k}\trans)\bvphi_{\bbeta_1,i} + \bu_{\bgam,k}\trans\bvphi_{\bgam,i} + o_p(1),
\end{eq}
with $\bv_{\bbeta_1,k}=\bzero$ when the PS model $\pi(\bx;\balph)$ is correctly specified and $\bu_{pa,\pi,k}=\bzero$ when either the PS model $\pi(\bx;\balph)$ or imputation model $g_{\xi}(\bgam\trans\bzpi)$ without the utility covariate is correctly specified.  Here $\bvphi_{\bbeta_1,i}$ and $\bvphi_{\bgam,i}$ are influence functions for $\bbetahat_1$ and $\bgamhat$ such that $n^{1/2}(\bbetahat_1 - \bbetabar_1) = n^{-1/2}\sum_{i=1}^n \bvphi_{\bbeta_1,i} + o_p(1)$ and $n^{1/2}(\bgamhat - \bgambar) = n^{-1/2}\sum_{i=1}^n \bvphi_{\bgam,i} + o_p(1)$.  Accordingly, the first term in \eqref{e:IF} represents the contribution from estimating $\bbeta_1$ in the baseline outcome model $\mu_k(\bx;\bbeta)$ for the double-index PS appearing in the IPW weight and the utility covariate.  The remaining term represents the contribution from estimating $\bgam$ in the imputation model $g_{\xi}(\bgam\trans\bzpi)$.  The influence function does not include terms associated with the variability in estimating $\balph$ in the parametric PS or for smoothing in the double-index PS, as such contributions to the expansion are of higher order when $N\gg n$ in the SS setting.  

When the PS model $\pi(\bx;\balph)$ is correctly specified, 
the influence function \eqref{e:IF} simplifies to: 
\begin{eq*}
n^{1/2}(\Delthat - \Delta) = n^{-1/2}\sum_{i=1}^n (\bu_{\bgam,1}-\bu_{\bgam,0})\trans\bvphi_{\bgam,i} + o_p(1),
\end{eq*}
where:
\begin{eq*}
\bvphi_{\bgam,i}=\left[\E\left\{\bZpii\bZpii\gdot_{\xi}(\bgambar\trans\bZpii)\right\}\right]^{-1}\bZpii\left\{Y_i - g_{\xi}(\bgambar\trans\bZpii)\right\},
\end{eq*}
for $\gdot_{\xi}(u) = \ddu g_{\xi}(u)\Big|_{u=u}$.
The centering of $Y_i$ around a model approximation of $\xi_T(\bV)$ suggests $\Delthat$ achieves efficiency gain over complete-case (CC) estimators, which neglect surrogates $\bW$.

\subsection{Efficiency Considerations} \label{ss:efficiencyconsiderations}
More formally, semiparametric efficiency theory establishes efficiency bounds for regular estimators of parameters of 
interest under specified semiparametric models \citep{bickel1998efficient}. Much of existing work involving semiparametric efficiency focus on
data with iid observations.  However, in our setup the observations in $\Dscr$ are not exactly iid as $L$ are constrained
such that $\sum_{i=1}^N L_i = n$ for a fixed $n$ and $\nu_n=n/N \to 0$ as $n\to \infty$. Moreover, because the distribution for 
the full data varies with $N$, as $\nu_n \to 0$, the very notion of a regular estimator is not clearly defined in this 
context. These issues complicate conventional applications of efficiency theory in the SS setting.

Let $\Deltabarstr= \E\{ \mu_1(\bX)-\mu_0(\bX)\}$ be strictly a functional of the observed data distribution not depending on identification assumptions \eqref{e:assump}-\eqref{e:assumpc}, as in $\Delta$.
Instead of directly calculating what would be the efficient influence function for $\Deltabarstr$ under the model for the 
full data, we take an alternative approach in which we calculate the efficient influence function for $\Deltabarstr$ under 
an \emph{ideal SS model} $\Mscr_{SS}$ for the labeled data $\Lscr$, which allows the conditional distribution of $Y\mid \bV,T$ 
to be unrestricted but assumes the distribution of $(\bV\trans,T)\trans$ is 
completely known. $\Mscr_{SS}$ approximates the SS setting, where the size of the unlabeled data $\Uscr$  
is much larger than that of $\Lscr$. As data in $\Lscr$ itself is iid with a fixed distribution, restricting the focus to $\Lscr$ avoids the
complications described above. We lastly calculate the asymptotic variance of $\Delthat$ through its 
associated influence function and compare against the efficiency bound under $\Mscr_{SS}$ to better understand the efficiency
of $\Delthat$ in context.

Following this approach, we show in Web Appendix C that the semiparametric efficiency bound for $\Deltabarstr$ under
$\Mscr_{SS}$, with respect to a class of regular parametric submodels subject to mild regularity conditions, is $\E(\varphi_{eff}^2)$, where:
\begin{eq}
    \varphi_{eff} = U_{\pi}\{ Y - \xi_T(\bV)\}
\end{eq}
is the efficient influence function.  This efficiency bound is lower than or equal to the efficiency bound in the fully 
nonparametric model where the distribution of $(Y,\bV\trans,T)\trans$ is unknown.  Furthermore, 
we show in Web Appendix C
that $\Delthat$ indeed achieves the SS efficiency bound when both the PS and imputation model, $\pi(\bx;\balph)$ and 
$g_{\xi}(\bgam\trans\bzpi)$, are correctly specified so that $\Delthat$ is locally semiparametric efficient.  Even though 
the distribution of $(\bV\trans,T)\trans$ is actually not known in our data setup, the bound under the ideal model $\Mscr_{SS}$ 
can still be achieved because $N \gg n$.  The correct specification of $\mu_k(\bx;\bbeta)$ is not required for attaining the efficiency bound, as the bound does not involve $\mu_k(\bx)$, but its specification is still important for double-robustness in case $\pi(\bx;\balph)$ is misspecified.  

The local efficiency of $\Delthat$ may prompt efficiency gains over non-efficient CC and SS estimators, though other estimators
can also achieve this efficiency bound.  For example, we show in Web Appendix C that a modification of the AIPW estimator 
from \cite{davidian2005semiparametric} also achieves the bound in our data setting, when its underlying working PS and
imputation models are correctly specified.  However, its effiency may still differ with that of $\Delthat$ under misspecification
of the working imputation model, as we consider in the simulations below.  Beyond these asymptotic properties under ideal 
conditions, we find through the simulations that $\Delthat$ can still achieve substantial efficiency gains over other estimators 
in finite samples under misspecified working models.  Another issue related to misspecificiation is that $\mu_k(\bx;\bbeta)$ 
and $g(\bgam\trans\bzpi)$ may not be compatible with one another if non-linear models such as logistic 
regression are used for either working model.  Using flexible basis expansion functions in $g_{\xi}(\bgam\trans\bzpi)$ to 
more closely approximate $\xi_k(\bv)$, as suggested in \ref{e:impmod}, can mitigate such incompatibility.  We find that
exact compatibility of the model is not a crucial requirement for $\Delthat$ to exhibit good performance in practice when
the working models are sufficiently flexible.
We next consider inference about $\Delta$ through a perturbation resampling procedure.

\subsection{Perturbation Resampling}
Although the asymptotic variance of $\Delthat$ can be obtained from the influence function in \eqref{e:IF}, a direct estimate is difficult because it would require estimating complicated functionals of the data distribution.  We instead propose a simple perturbation resampling procedure for inference based on \cite{jin2001simple}.  Let $\Gscr = \{ G_i : i=1,\ldots,N\}$ be non-negative iid random variables with unit mean and variance that are independent of the observed data $\Dscr$.  We first obtained perturbed estimators of $\balph$  and $\bbeta$:
\begin{eq*}
&\balphhat^* = \argmin_{\balph} \left\{ -n^{-1}\sum_{i=1}^n \ell_{\pi}(\balph;\bX_i,T_i)G_i + p_{\lambda_N}^*(\balph_1)\right\} \\
&\bbetahat^* = \argmin_{\bbeta} \left\{ -n^{-1}\sum_{i=1}^n \ell_{\mu}(\bbeta;Y_i,\bX_i,T_i)G_i + p_{\lambda_n}^*(\bbeta_{\{-1\}})\right\},
\end{eq*}
where $p_{\lambda_N}^*(\cdot)$ and $p_{\lambda_n}^*(\cdot)$ are the corresponding penalties based weights estimated by the same perturbation procedure if data-adaptive weights are used, as in adaptive LASSO.  This leads to the perturbed DiPS estimate:
\begin{eq*}
\pihat_1^*(\bx;\balphhat_1^*,\bbetahat_1^*)= \frac{\sum_{j=1}^N K_h(\bShat_j^* - \bshat^*)I(T_j=1) G_j}{\sum_{j=1}^N K_h(\bShat_j^* - \bshat^*)G_j},
\end{eq*}
where $\bShat_j = (\balphhat_1^*,\bbetahat_1^*)\trans\bX_j$ and $\bshat^* = (\balphhat_1^*,\bbetahat_1^*)\trans\bx$ are the perturbed bivariate scores. We then obtain the perturbed estimator $\bgamhat^*$ as the solution to:
\begin{eq*}
n^{-1}\sum_{i=1}^n \bZpihatstri \left\{ Y_i - g_{\xi}(\bgam\trans\bZpihatstri)\right\}G_i + \lambda_n\bgam_{\{-1\}} = \bzero,
\end{eq*}
where $\pihat^*$ specifies that the imputations use utility covariates that plug in $\pihat^*(\bx;\balphhat_1^*,\bbetahat_1^*)$. 
Finally, we calculate the perturbed $\SSDR$ estimator as $\Delthat^* = \muhat^*_1-\muhat^*_0$, where:
\begin{eq*}
&\muhat^*_1 = \left\{ \sum_{i=1}^{\ntot}\frac{I(T_i=1)G_i}{\pihat(\bX_i;\balphhat_1^*,\bbetahat_1^*)}\right\}^{-1}\left\{ \sum_{i=1}^{\ntot}\frac{I(T_i=1)\Ystrhat_i G_i}{\pihat(\bX_i;\balphhat_1^*,\bbetahat_1^*)}\right\} \\
\text{and }&\muhat^*_0 = \left\{ \sum_{i=1}^{\ntot}\frac{I(T_i=0)G_i}{1-\pihat(\bX_i;\balphhat_1^*,\bbetahat_1^*)}\right\}^{-1}\left\{ \sum_{i=0}^{\ntot}\frac{I(T_i=0)\Ystrhat_i G_i}{1-\pihat(\bX_i;\balphhat_1^*,\bbetahat_1^*)}\right\}, 
\end{eq*}
with $\Ystrhat_i = g_{\xi}(\bgamhat^{*\sf \tiny T}\bZpihatstri)$.  
It can be shown using arguments from \cite{jin2001simple} that the asymptotic distribution of $n^{1/2}(\Delthat - \Deltbar)$ coincides with that of $n^{1/2}(\Delthat^* - \Delthat)\mid \Dscr$.  
In this scheme, estimation of the initial PS model $\balphhat^*$, the DiPS $\pihat^*(\bx;\balphhat_1^*,\bbetahat_1^*)$, and the final IPW estimators $\muhat_1^*$ 
and $\muhat_0^*$ does not technically need to be perturbed as they are estimated based on data from $\Uscr$.
Consequently their contributions to the asymptotic variance is of higher order when $N \gg n$. However, we found that not perturbing these steps can have some impact on the standard error estimation in finite samples if $N$ is not yet very large relative to $n$ and perturbed these steps by default. We approximate the standard error of $\Delthat$ based on the empirical standard deviation, or, as a robust alternative, the mean absolute deviation (MAD), of a large number of samples of $\Delthat^*$ and construct confidence intervals (CI) based on the empirical percentiles.
\section{Simulations}
\label{s:sims}

We assessed through simulations the finite samples bias, standard error (SE), root mean square error (RMSE), and relative efficiency (RE) of our proposed estimator ($\SSDR$) compared to alternative estimators.  In separate simulations we also examined the performance of the perturbation procedure for inference based on $\SSDR$.  For $\SSDR$, we specified $\bh(\cdot)$ in the imputation model in \eqref{e:impmod} as natural cubic splines with 6 knots specified at uniform quantiles.  
Ridge regression with the tuning parameter chosen by cross-validation on the deviance was used for regularization in \eqref{e:gamest}.
Although it is unclear whether cross-validation for choosing the tuning parameter satisfies with high probability the 
$\lambda_n = o(n^{-1/2})$ condition, we found that
it leads to good performance in finite samples and that the simulation results were not sensitive to the choice of the tuning parameter (Web Appendix D).
Adaptive LASSO with initial weights estimated by ridge regression and tuning parameter chosen by minimizing a modified BIC criteria \citep{minnier2011perturbation} 
was used for estimating $\balphhat$ and $\bbetahat$.  
A plug-in estimate was used for the bandwidth in the smoothing for the double-index PS \citep{cheng2017estimating}.  Prior to smoothing the components of $\bShat$ were standardized and transformed by a probability integral transform based on the normal cumulative distribution function to induce approximately uniformly distributed inputs, which can improve finite-sample performance \citep{wand1991transformations}.  As we focused on binary outcomes, we specified logistic link functions $g_{\xi}(u)=g_{\pi}(u)=g_{\mu}(u)=1/(1+e^{-u})$ for the working models in \eqref{e:impmod} and \eqref{e:orwrk}.  

For comparison, we considered the standard complete-case AIPW estimator \citep{lunceford2004stratification}
based on $\Lscr$ only ($\CCDR$). We also considered the estimator from \cite{davidian2005semiparametric}, which we adapted to our setting by replacing
instances of the treatment randomization probability with PS estimates $\pi(\bX_i;\balphhat)$ ($\SSPrePost$).  $\SSPrePost$ 
also leverages the surrogate data in $\Uscr$ alongside labeled data $\Lscr$ to facilitate estimation of the ATE and 
is locally efficient, as discussed in Section \ref{ss:efficiencyconsiderations}.  It is also DR in that 
it is consistent for $\Delta$ if either the PS model $\pi(\bx;\balph)$ or outcome regression model $\mu_k(\bx;\bbeta)$ is correctly specified.

To mimic the EHR data, we considered the case with $Y$ as binary and $\bW$ as count variables.  In all scenarios, data were generated according to $\bX\sim N\left\{\bzero,\sigma^2_x (1-\rho_x)\bI+\sigma^2_x\rho_x\right\}$, $T\mid\bX\sim Bernoulli\left\{\pi(\bX)\right\}$, $Y\mid\bX,T \sim Bernoulli\{ \mu_T(\bX)\}$, and $\bW=\lfloor\bGam (1,\bX\trans,T,Y)\trans + \beps\rfloor$, where 
$\beps \sim N\{\bzero, \sigma^2_w (1-\rho_w)\bI + \sigma^2_w\rho_w\}$ and $\lfloor \cdot \rfloor$ is the floor function.  
Initially we considered simulating data that roughly resembled the EHR data example in terms of model parameters but
were simple enough to be more broadly relevant. To this end we considered $p_x = 10$ baseline covariates and $p_w=5$ surrogates, with variances and correlations $\sigma^2_x = 1$, $\rho_x = .2$, $\sigma^2_w = 5$, $\rho_w = .2$, and
$\bGam_{5\times 13} = (\bzero_{5\times 1}, .1\bone_{5\times 5}, -.1\bone_{5\times 5}, .1\bone_{5\times 1},\mathbf{g})$.  These simulations were varied over different model specifications, predictive strength of surrogates $\bW$, and sample sizes.  The imputation model $\xi(\bv;\bgam)$ was misspecified throughout, 
and either the PS model $\pi(\bx;\balph)$ or baseline outcome model $\mu_k(\bx;\bbeta)$ were potentially misspecified
by setting the true $\pi(\bx)$ and $\mu_k(\bx)$ to be double-index models that allow for pairwise interactions:
\begin{eq*}
&(1) \textit{Both correct: } \mu_k(\bx) = g_{\mu}(\beta_0 + \bbeta_1\trans\bx + \beta_2 k), \quad \pi(\bx) = g_{\pi}(\alpha_0 + \balph_1\trans\bx) \\
&(2) \textit{Misspecified $\mu$: } \mu_k(\bx) = g_{\mu}\left\{\beta_0 + \bbeta_{1[1]}\trans\bx(\bbeta_{1[2]}\trans\bx+1) + \beta_2 k\right\}, \quad \pi(\bx) = g_{\pi}(\alpha_0 + \balph_1\trans\bx) \\
&(3) \textit{Misspecified $\pi$: } \mu_k(\bx) = g_{\mu}(\beta_0 + \bbeta_1\trans\bx + \beta_2 k), \quad \pi(\bx) = g_{\pi}\left\{\alpha_0 + \balph_{1[1]}\trans\bx (\balph_{1[2]}\trans\bx+1)\right\},
\end{eq*}
where $g_{\mu}(u)=g_{\pi}(u)=1/(1+e^{-u})$ and parameter values were:
\begin{eq*}
&\alpha_0 = -.3, \quad\balph_1 = .35 \bone_{1\times 10}, \quad\beta_0 = -.65, \quad\bbeta_1 = (1\bone_{1\times 3}, .5\bone_{1\times 3},-1.15,-1\bone_{1\times 3})\trans\\
&\balph_{1[1]} = .5(0,.35,0,.35,0,.35,0,.35,0,.35)\trans, \quad\balph_{1[2]} = (.35,0,.35,0,.35,0,.35,0,.35,0)\trans \\
&\bbeta_{1[1]}=.5(1,0,1,0,.5,0,-.5,0,-1,0)\trans, \quad \bbeta_{1[2]}=(0,.5,0,.5,0,.5,0,.5,0,.5)\trans.
\end{eq*}
We considered $(2.1,2.1,1,0,0)\trans$, $(5,5,2.5,0,0)\trans$, and $(15,10,5,2.5,0)\trans$ for $\mathbf{g}$ to model cases where $\bW$ has mild, moderate, and strong predictive strength, respectively. 
The outcome prevalence was approximately 45\% in all scenarios, and the treatment prevalence was also 45\%, except in the misspecified $\pi$ scenario where it was approximately 53\%.
Sample sizes were varied over $n=100,250,500$ and $N=500,1000,5000,10000$ in a factorial design to show separate effects of increasing $n$ and $N$.  
To more closely approximate the data in the EHR data application, in a separate simulation, we generated data under the ``both correct'' setting 
using coefficient estimates from fitting corresponding working models in the EHR data as the values set for $\balph,\bbeta,\bGam$. For these simulations
we only considered the size $n=100, N=1000$.
The results in each scenario are summarized from 1,000 simulated datasets. 

Table \ref{tab:bias} presents the bias, SE, and RMSE across misspecification scenarios with $\bW$ at moderate strength.  $\SSDR$, $\SSPrePost$, and $\CCDR$ exhibits low bias that diminishes to zero as sample size increased under all three scenarios, verifying their double-robustness.  
Decreases in bias for $\SSDR$ appears to be largely driven by increasing $n$.
Figure \ref{fig:sev} presents the RE under the correctly specified scenario varying the strength of association between $\bW$ and $Y$.
Generally, $\SSDR$ uniformly achieves the lowest RMSE for a given $n$.  Increasing $N$ while fixing $n$ improves the RMSE for $\SSDR$ and $\SSPrePost$,
but the improvements are limited by $n$, which drives the asymptotic variance in the SS regime as shown in the asymptotic analysis.
The benefit of additional unlabeled data varies with the strength of $\bW$.  The RMSE for $\SSDR$ does not improve much with
larger $N$ for a fixed $n$ when $\bW$ is weakly correlated with $Y$ but improves greatly when $\bW$ are strongly correlated.

\begin{table}[ht]
    \caption{Bias, SE, and RMSE of estimators under different model misspecification scenarios over 1,000 simulated datasets.  
}
\label{tab:bias}%
    \centering
    \scalebox{0.8}{
    \begin{tabular}{cccccccccccc}
        \toprule
        \multicolumn{3}{c}{}& \multicolumn{3}{c}{Both Correct}& \multicolumn{3}{c}{Misspecified $\mu$}& \multicolumn{3}{c}{Misspecified $\pi$}\\
        $n$ & $N$ & \textbf{Estimator} & \textbf{Bias} & \textbf{SE} & \textbf{RMSE} & \textbf{Bias} & \textbf{SE} & \textbf{RMSE} & \textbf{Bias} & \textbf{SE} & \textbf{RMSE} \\ 
      \hline
    100 & 1000 & $\CCDR$ & -0.001 & 0.140 & 0.140 & -0.003 & 0.176 & 0.176 & 0.002 & 0.088 & 0.088 \\ 
      100 & 1000 & $\SSPrePost$ & -0.006 & 0.099 & 0.099 & -0.009 & 0.110 & 0.110 & -0.003 & 0.058 & 0.058 \\ 
      100 & 1000 & $\SSDR$ & -0.007 & 0.055 & 0.056 & -0.016 & 0.061 & 0.063 & -0.008 & 0.046 & 0.047 \\ \hline
      100 & 10000 & $\CCDR$ & 0.005 & 0.127 & 0.127 & 0.008 & 0.172 & 0.172 & 0.001 & 0.083 & 0.083 \\ 
      100 & 10000 & $\SSPrePost$ & -0.003 & 0.076 & 0.076 & -0.003 & 0.113 & 0.113 & -0.005 & 0.053 & 0.053 \\ 
      100 & 10000 & $\SSDR$ & -0.005 & 0.049 & 0.049 & -0.014 & 0.056 & 0.058 & -0.009 & 0.040 & 0.042 \\ \hline
      500 & 1000 & $\CCDR$ & -0.002 & 0.063 & 0.063 & -0.000 & 0.083 & 0.083 & 0.001 & 0.034 & 0.034 \\ 
      500 & 1000 & $\SSPrePost$ & -0.003 & 0.052 & 0.052 & -0.001 & 0.072 & 0.072 & -0.000 & 0.029 & 0.029 \\ 
      500 & 1000 & $\SSDR$ & -0.003 & 0.034 & 0.034 & -0.002 & 0.044 & 0.044 & -0.003 & 0.027 & 0.028 \\ \hline
      500 & 10000 & $\CCDR$ & -0.002 & 0.059 & 0.059 & 0.002 & 0.081 & 0.081 & -0.000 & 0.035 & 0.035 \\ 
      500 & 10000 & $\SSPrePost$ & -0.002 & 0.041 & 0.041 & -0.001 & 0.048 & 0.048 & -0.001 & 0.023 & 0.023 \\ 
      500 & 10000 & $\SSDR$ & -0.003 & 0.026 & 0.026 & -0.002 & 0.033 & 0.033 & -0.004 & 0.020 & 0.020 \\ 
       \hline
    \end{tabular}
    }
\end{table}

\begin{figure}[h!]
    \caption{
        Root mean square error (RMSE) of $\CCDR$, $\SSPrePost$, and $\SSDR$ by strength of the surrogates $\bW$ (Mild, Moderate, Strong) over 
        different sample sizes.
    }
    \label{fig:sev}
    \begin{center}
    \includegraphics[scale=.22]{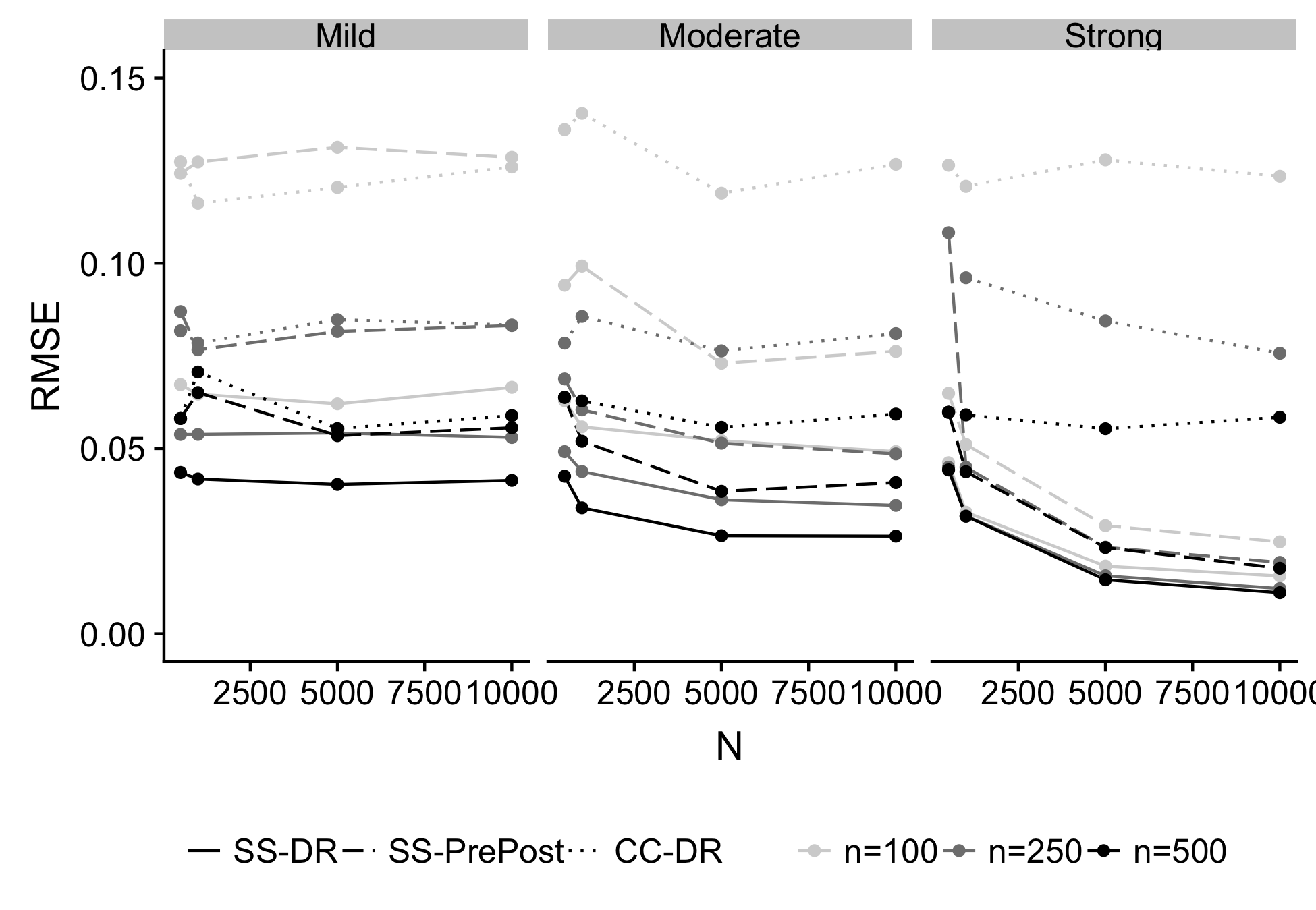}
    \end{center}
\end{figure}

Figure \ref{fig:re} presents the RE of various estimators relative to $\SSDR$ across misspecification scenarios with moderate $\bW$.  
$\SSDR$ is more efficient than both $\CCDR$ and $\SSPrePost$ across misspecification settings and sample sizes. It gains
over $\CCDR$ as it makes use of the unlabeled data $\Uscr$.  The gains over $\SSPrePost$ suggest that $\SSDR$ can be more efficient
under misspecification of the working imputation model.  In other simulations not presented we found that $\SSDR$ has
similar efficiency with $\SSPrePost$ under a correctly specified imputation model, as expected since both are locally efficient.  $\SSDR$ may also achieve efficiency gains
relative to other estimators involving PS weighting in finite samples when the true PS are more extreme. The calibrated
estimate $\pihat(\bx;\balphhat_1,\bbetahat_1)$ used in $\SSDR$ pulls estimates of the PS away from $0$ or $1$, which can 
lead to more stable final estimates when $\pihat(\bx;\balphhat_1,\bbetahat_1)$ is used in reweighting.  Lastly, $\SSDR$ 
may exhibit some efficiency gains over $\SSPrePost$ in finite samples from using regularization for estimating the nuisance 
parameters, whereas $\SSPrePost$ uses unregularized maximum likelihood estimators in our implementation.
In the data application scenario,
$\SSDR$ was significantly more efficient, having a RMSE of .03 compared to .15 and .12 for $\CCDR$ and $\SSPrePost$, suggesting
the strength of surrogates are strong in the EHR data.

\begin{figure}[h!]
\caption{RE of estimators, defined as the ratio of mean square errors (MSE) relative to $\SSDR$, by model misspecification scenarios for different sample sizes over 1,000 simulated datasets.  Higher values of RE denotes greater efficiency (lower MSE) relative to $\SSDR$.
}
\label{fig:re}
\begin{center}
\includegraphics[scale=.0635]{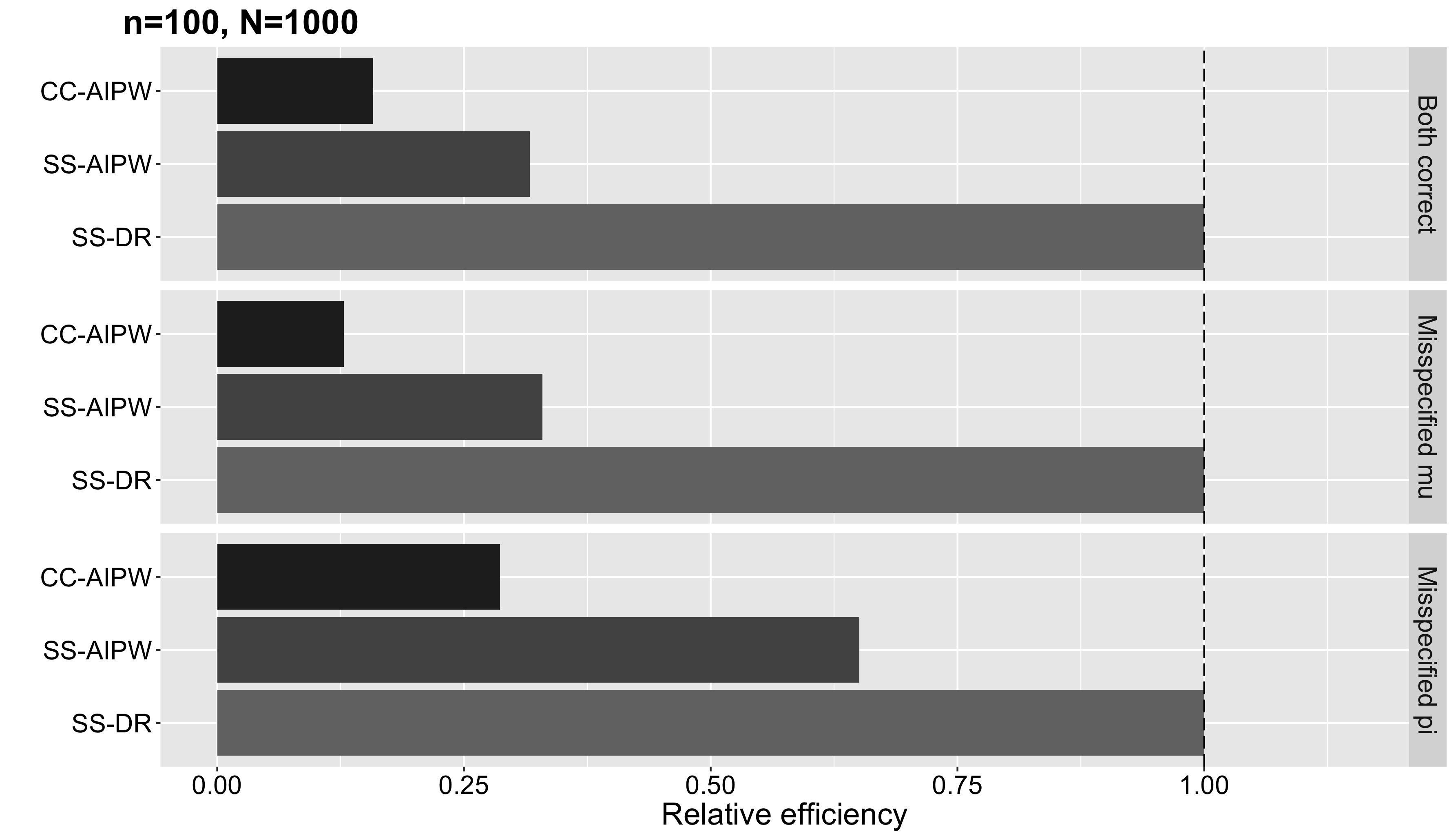} \includegraphics[scale=.0635]{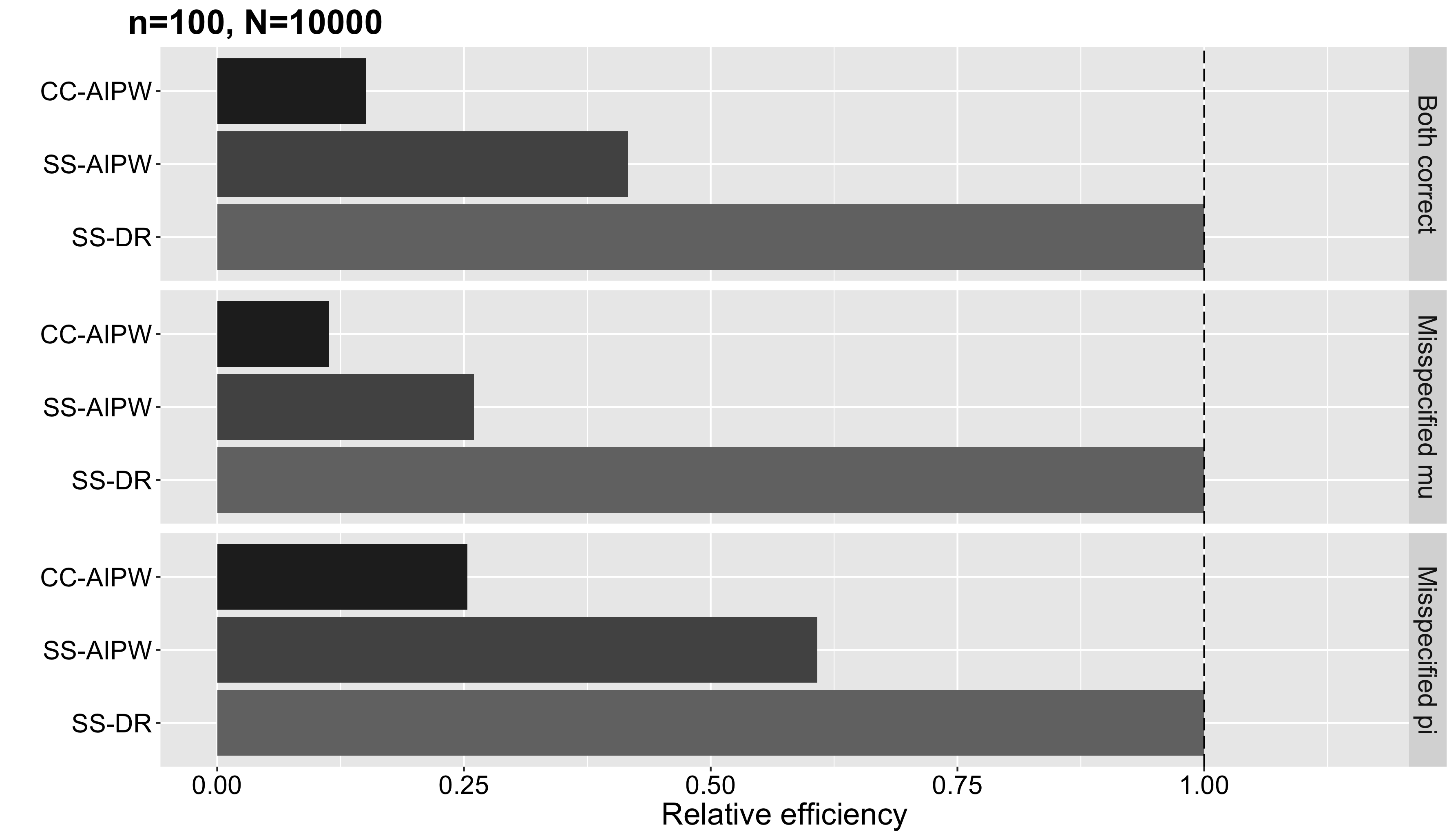}
\includegraphics[scale=.0635]{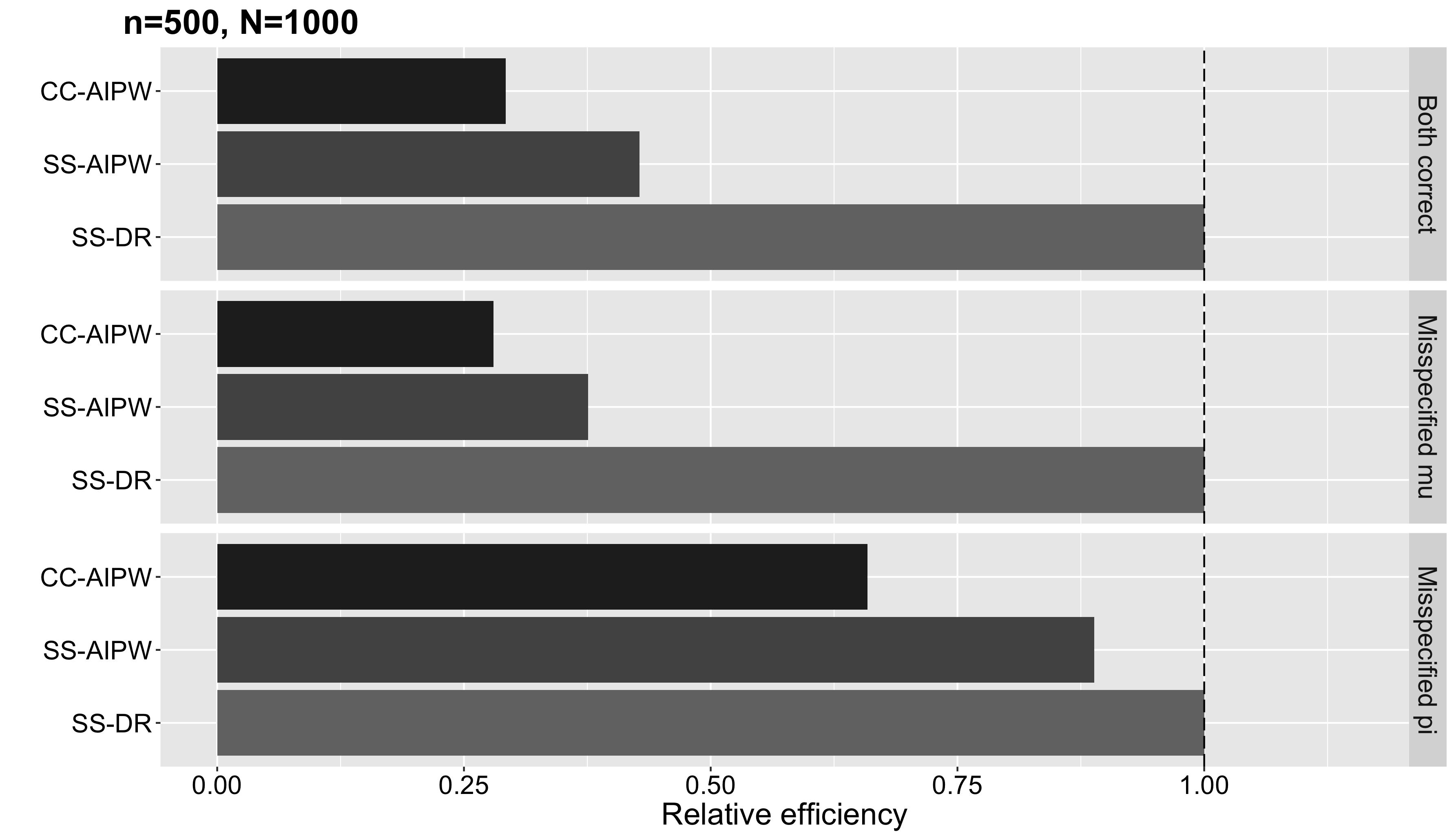} \includegraphics[scale=.0635]{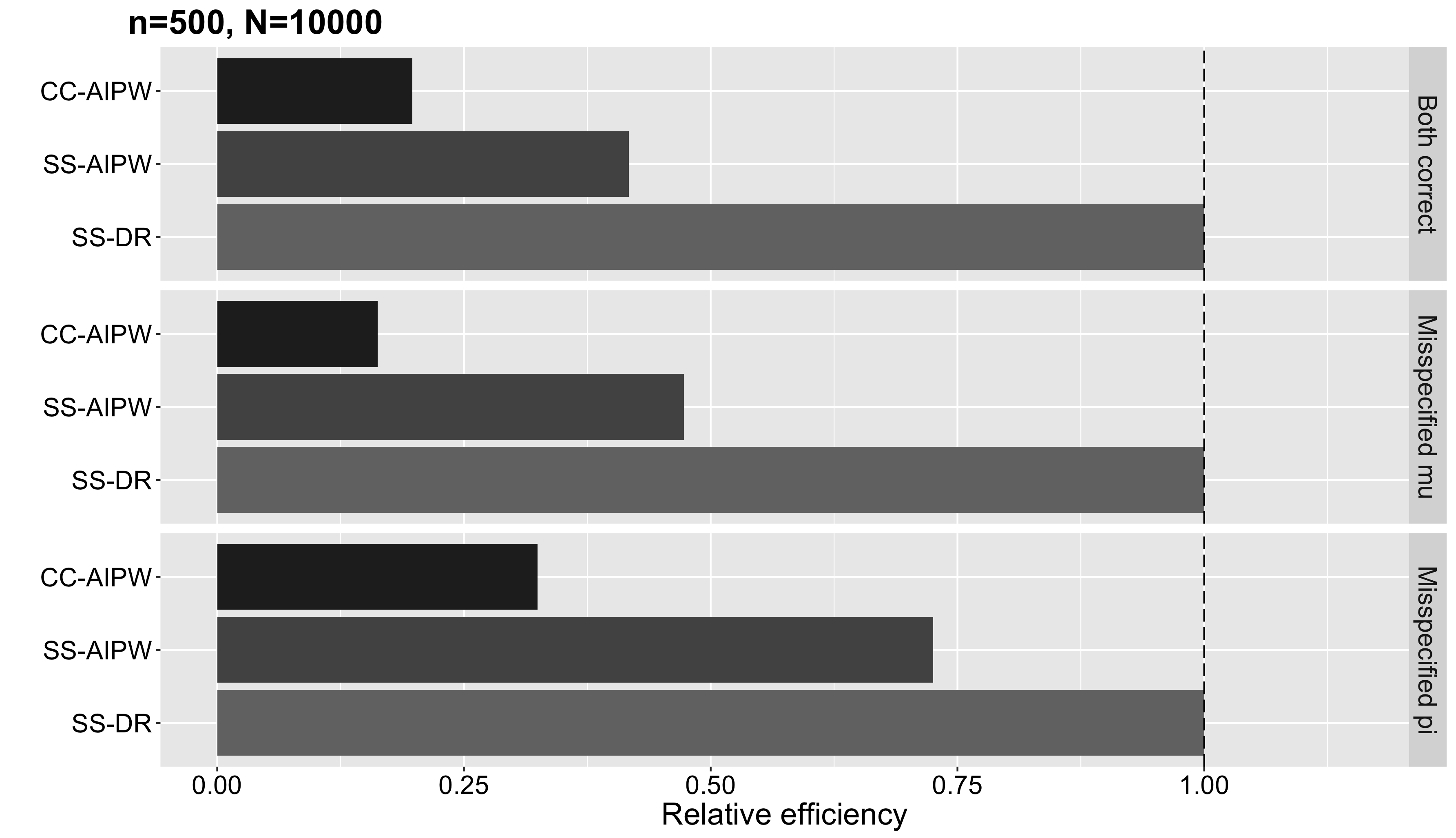}
\end{center}
\end{figure}

To implement the perturbation procedure, we used the weights $G_i \sim 4 \times Beta(.5,1.5)$ and 1,000 sets of $\Gscr$ for SE and CI estimation.  
We considered evaluating the perturbations only in the scenario where both $\mu_k(\bx;\bbeta)$ and $\pi(\bx;\balph)$ were correctly specified and $\bW$ had moderate predictive strength.  
The results are presented in Table \ref{tab:CIsim}. 
In both small and large samples, the SEs estimated by the standard deviation and by MAD approximated well the empirical 
standard error.  The coverage of the percentiles were also close to nominal levels, albeit slightly conservative.  Results from 
4 of the simulation iterations for when $n=100, N=1000$ and from 8 of the iterations when $n=250, N=5000$ were omitted as
the simulations timed out from prolonged computational time.

\begin{table}[h!]
    \caption{Performance of perturbation resampling for $\SSDR$ in 1,000 simulated datasets when both $\mu_k(\bx;\bbeta)$ and $\pi(\bx;\balph)$ are correctly specified.  Emp SE: empirical SE of $\SSDR$ over simulated datasets; ASE: average of estimated SE based on the standard deviation of perturbed estimates; ASE$_{\text{MAD}}$: average of SE based on MAD of perturbed estimates; RMSE: root-mean square error; Coverage: coverage of 95\% percentile CIs.}
    \label{tab:CIsim}
    \begin{center}
    \begin{tabular}{ccccccc}
        \toprule
        \textbf{Size} & \textbf{Bias} & \textbf{Emp SE} & \textbf{ASE} & \textbf{ASE$_{\text{MAD}}$} & \textbf{RMSE} & \textbf{Coverage} \\
        \midrule
        $n = 100$, $N=1000$ &  -0.002 & 0.055 & 0.052 & 0.052 & 0.055 & 0.963 \\
        $n = 250$, $N=5000$ & -0.002 & 0.034 & 0.034 & 0.034 & 0.034 & 0.963 \\ 
        \bottomrule
    \end{tabular}
    \end{center}
\end{table}

\section{EHR Data Application}
\label{s:data}
We applied $\SSDR$ and the alternative estimators to compare the rates of treatment response to two biologic agents for treating inflammatory bowel disease (IBD) using the EMRs from Partner's Healthcare.  Though the efficacy and effectiveness of adalimumab (ADA) and infliximab (IFX) for  the management of IBD have been established individually, few studies have offered a direct comparison. Consequently the choice of treatment in practice is often influenced by factors other than comparative performance \citep{ananthakrishnan2016comparative}.  Randomized trials may be unfeasible due to the large number of patients that would be needed to detect the presumed small treatment difference, and other observational data lack detailed clinical information needed to ascertain meaningful outcomes.  EHRs are thus uniquely positioned to provide evidence on the comparative effectiveness of these two therapies.  

The data we considered consisted of $N=1,243$ total IBD patients, including 200 who initiated treatment with ADA and 1043 with IFX.  Through chart review by a gastroenterologist, a random subset of $n=117$ records were labeled with the true treatment response status (responder vs. non-responder) within one year of treatment initiation.  We included 12 baseline covariates to adjust for confounding in $\bX$, including demographics, comorbidities, prior utilization, and inflammation biomarker levels.  We also selected 35 post-treatment surrogates for $\bW$, comprising of counts of NLP mentions of clinically relevant terms (e.g. ``bleeding'', ``fistula'', ``tenesmus'') within one year of initiation.  The transformation $u \mapsto log(1+u)$ was applied to all count variables in $\bV$ to mitigate instability in the estimation due to skewness in their distributions.  Nonparametric bootstrap was used to estimate SEs and CIs for the alternative estimators and perturbation for $\SSDR$, using the MAD of resampled estimates as an robust estimator of the SEs.  In addition we calculated two-sided p-values based on inverting percentile CIs from the resampled estimates, using the equivalence between significance tests and confidence sets \citep{liu1997notions}.

As shown in Table \ref{tab:data}, the point estimates of most estimators agreed that patients receiving ADA experienced 
lower rates of treatment response, after adjustment for confounding.  $\SSDR$ is estimated to achieve more than 600\% 
efficiency gain over CC estimators and 450\% efficiency gain over the other SS estimators based on the estimated variances.  
It is the only estimator that exhibits a difference that is significant at the .05 level, suggesting that 
patients receiving IFX experience a slightly higher rate of response to treatment.

\begin{table}[h!]
\caption{Point and SE estimates based on MAD for the ATE of ADA vs. IFX, with respect to one-year treatment response rate, among IBD patients in EMR data based on various methods, including the naive CC estimator ($\CCNaive$) that completely ignores confounding bias. 
95\% CIs are percentile-based CIs from resampling and p-values are for testing $H_0: \Delta=0$ based on inverting percentile CIs.}
\label{tab:data}
\begin{center}
\begin{tabular}{lcccc}
    \toprule
    \textbf{Estimator} & \textbf{Estimate} & \textbf{SE} & \textbf{95\% CI (Pct)} & \textbf{p-value} \\
    \midrule
    $\CCNaive$ & 0.014 & 0.099 & (-0.201, 0.177) & 0.822 \\
    $\CCDR$ & -0.125 & 0.153 & (-0.416, 0.164) & 0.592 \\
    $\SSPrePost$ & 0.033 & 0.109 & (-0.265, 0.180) & 0.778 \\
    $\SSDR$ & -0.067 & 0.036 & (-0.164, -0.002) & 0.044 \\
    \bottomrule
\end{tabular}
\end{center}
\end{table}

\section{Discussion}
\label{s:discuss}

This paper developed a robust and efficient estimator for the ATE in a SS setting where the true outcome is labeled
for a vanishingly small proportion of the entire set of observations.
The estimator adopts an imputation approach to leverage surrogate data from $\Uscr$ to improve efficiency that is robust to
misspecification of the imputation model.  
It is DR, locally semiparametric efficient under an ideal SS semiparametric model, and demonstrated to be more efficient
than CC and other estimators that leverage $\Uscr$ in finite samples.

We have assumed that the true outcomes $Y$ are labeled completely at random, which may be reasonable if investigators control the labeling.  But this assumption could be restrictive if labeling was stratified by some known factors or if some records that are available were not labeled for research purposes.
One possible approach to address the case where $Y$ are missing at random is to apply weighting or semiparametric efficient methods \citep{robins1994estimation} to the estimating equation when estimating $\bgam$ in \eqref{e:gamest}.  
Other refinements to our proposed approach are possible.  For example, in the case where $\bW$ is high dimensional, the 
group LASSO  \citep{yuan2006model}, where the basis expansion functions for each  surrogate variable are grouped together, 
can also potentially be used to improve efficiency in finite-samples.  It would also be of interest to extend the theoretical results to the case where $p_x$ and $p_w$ are allowed to diverge with $n$.

\backmatter

\section*{Acknowledgements}

The authors thank Ray Liu, Eric Tchetgen Tchetgen, Rajarshi Mukherjee, and James Robins for helpful discussions as well as 
the editor, associate editor and two referees for their insightful feedback and suggestions.  
Much of this work was done when the first author was a graduate student at Harvard University.  This work was supported 
by National Institutes of Health grants T32CA009337, R21CA242940, and R01HL089778.
The views expressed in this article are those of the authors and do not necessarily reflect the views of the Department of Veterans Affairs.
\vspace*{-8pt}

\bibliographystyle{biom} 
\bibliography{mybibilo}

\section*{Supporting Information}

Web Appendices referenced in Sections \ref{s:method} and \ref{s:sims} are available with this paper at the 
Biometrics website on Wiley Online Library.\vspace*{-8pt}

\label{lastpage}

\processdelayedfloats

\newpage

\section*{Supporting Information for ``Robust and Efficient Semi-Supervised Estimation of Average Treatment Effects with Application to Electronic Health Records Data'' by David Cheng, Ashwin N. Ananthakrishnan, and Tianxi Cai}

\newpage

In the following, the supporting lemmas of Web Appendix A identify rates of convergence for frequently encountered quantities and also identify the efficient influence function for $\Deltabarstr$ under a fully nonparametric model.  Web Appendix A also sketches the proof that $\bgamhat$ is $n^{1/2}-$consistent. The results in Web Appendix B show that $\Delthat$ is consistent and asymptotically linear, deriving its influence function. The results in Web Appendix C establish the semiparametric efficiency bound under the SS model and shows that $\Delthat$ and a modified version of the AIPW estimator proposed in \cite{davidian2005semiparametric} achieves this bound at particular distributions for the data so that it is locally semiparametric efficient. Finally, Web Appendix D reports results 
from brief simulations to gauge the impact of the choice of tuning parameter used in ridge regression on $\Delthat$ in finite samples. Throughout Web Appendices A-C we assume that mild regularity conditions required for the double-index PS in Web Appendix A of \cite{cheng2017estimating} hold.

The following notations facilitate the subsequent derivations.  Let $\pi_k(\bx)=\P(T=k\mid \bX=\bx)$ for $k=0,1$.  Let $\pi_k(\bx;\balph_1,\bbeta_1)=\P(T=k\mid\balph_1\trans\bX=\balph_1\trans\bx,\bbeta_1\trans\bX=\bbeta_1\trans\bx)$, $\pi_k(\bx;\balph)=\pi(\bx;\balph)^k\{ 1-\pi(\bx;\balph)\}^{1-k}$, and $\pihat_k(\bx;\balph_1,\bbeta_1)=\pihat(\bx;\balph_1,\bbeta_1)^k\{1-\pihat_k(\bx;\balph_1,\bbeta_1)\}^{1-k}$ for given $\balph_1,\bbeta_1\in\mathbb{R}^p$ and $\balph\in\mathbb{R}^{p+1}$ and $k=0,1$. Moreover, let $\bvthetbar = (\balphbar_1\trans,\bbetabar_1\trans)\trans$, $\bvthethat = (\balphhat_1\trans,\bbetahat_1\trans)\trans$, $\pi_k(\bx;\bvthetbar)=\pi_k(\bx;\balphbar_1,\bbetabar_1)$, $\pihat_k(\bx;\bvthetbar)= \pihat_k(\bx;\balphbar_1,\bbetabar_1)$, and $\pihat_k(\bx;\bvthethat)=\pihat_k(\bx;\balphhat_1,\bbetahat_1)$. Let the working imputation model be denoted by $\xi_T(\bV;\bgam,\pi)=g_{\xi}\{\bgam\trans(1,\bh(\bV)\trans,T,U_{\pi})\trans\}$, where $U_{\pi}= I(T=1)/\pi(\bX) - I(T=0)/\{ 1-\pi(\bX)\}$, given some PS $\pi$ .  Let $\omega_{k,i} = I(T_i = k)/\pi_k(\bX_i)$, $\omegabar_{k,i} = I(T_i=k)/\pi_k(\bX_i,\bvthetbar)$, $\omegahat_{k,i} = I(T_i = k)/\pihat_k(\bX_i,\bvthethat)$, and $\bSbar = (\balphbar_1,\bbetabar_1)\trans\bX$ with $\bSbar_i = (\balphbar_1,\bbetabar_1)\trans\bX_i$, for $k=0,1$ and $i=1,\ldots,N$. 

\section*{Web Appendix A: Supporting Lemmas}
\begin{lemma}
\label{l:0}
The rates of uniform convergence for kernel estimators we use are as follows:
\begin{eq}
&\sup_{\bx}\norm{\pihat_k(\bx;\bvthetbar)-\pi_k(\bx;\bvthetbar)}= O_p(\atld_N) \\
&\sup_{\bx}\norm{\ddbalph\pihat_k(\bx;\bvthetbar)-\ddbalph\pi_k(\bx;\bvthetbar)} = O_p(\btld_N) \\
&\sup_{\bx}\norm{\ddbbeta\pihat_k(\bx;\bvthetbar)-\ddbbeta\pi_k(\bx;\bvthetbar)} = O_p(\btld_N) \\
&\sup_{\bx}\norm{\pihat_k(\bx;\bvthethat)-\pi_k(\bx;\bvthetbar)}=O_p(a_n),
\end{eq}
where:
\begin{eq*}
\atld_N = h^q + \left(\frac{logN}{Nh^2}\right)^{\half} \text{, } \btld_N = h^{q} + \left(\frac{logN}{Nh^{4}}\right)^{\half} \text{, and } a_n = n^{-\half} + n^{-\half}\btld_N + \atld_N.
\end{eq*}
\begin{proof}
The uniform rates for fixed $\balphbar_1$ and $\bbetabar_1$ in the first three equations follow directly from the uniform convergence rates of kernel smoothers and their first derivatives \citep{hansen2008uniform}.  To establish the uniform convergence rate for DiPS, we first note that:
\begin{eq*}
\sup_{\bx}\norm{\pihat_k(\bx;\bvthethat)-\pi_k(\bx;\bvthetbar)} &\leq \sup_{\bx}\norm{\pihat_k(\bx;\bvthethat) - \pihat_k(\bx;\bvthetbar)} +\sup_{\bx}\norm{\pihat_k(\bx;\bvthetbar)-\pi_k(\bX;\bvthetbar)}.
\end{eq*}
The first term on the right-hand side can be written:
\begin{eq*}
&\sup_{\bx}\norm{\pihat_k(\bx;\bvthethat) - \pihat_k(\bx;\bvthetbar)} \\
&\qquad\leq \sup_{\bx}\norm{\ddbalph\pi_k(\bx;\balphbar_1,\bbetabar_1)(\balphhat_1 -\balphbar_1) + \ddbbeta\pi_k(\bx;\balphbar_1,\bbetabar_1)(\bbetahat_1-\bbetabar_1)} \\
&\qquad\qquad + \sup_{\bx}\norm{\left\{\ddbalph\pihat_k(\bx;\balphbar_1,\bbetabar_1)-\ddbalph\pi_k(\bx;\balphbar_1,\bbetabar_1)\right\}(\balphhat_1-\balphbar_1)} \\
&\qquad\qquad + \sup_{\bx}\norm{\left\{\ddbbeta\pihat_k(\bx;\balphbar_1,\bbetabar_1)-\ddbbeta\pi_k(\bx;\balphbar_1,\bbetabar_1)\right\}(\bbetahat_1-\bbetabar_1)} \\
&\qquad\qquad + O_p(\norm{\balphhat_1-\balphbar_1}^2+\norm{\bbetahat_1-\bbetabar_1}^2 + \norm{\balphhat_1 - \balphbar_1}\norm{\bbetahat_1 - \bbetabar_1}).
\end{eq*}
We obtain the desired rate by collecting terms and applying the other rates from above, using that $\ddbalph\pi_k(\bx;\balphbar_1,\bbetabar_1)$ and $\ddbbeta\pi_k(\bx;\balphbar_1,\bbetabar_1)$ are continuous in $\bx$, and $\bx$ lies in a compact covariate space.
\end{proof}
\end{lemma}
\begin{lemma}
\label{l:1}
Let $\zeta_i=g(\bZ_i)$ be some integrable function of $\bZ_i=(\bV_i\trans,T_i)\trans$, for $i=1,\ldots,N$.  Then:
\begin{eq}
\ntotinv \sum_{i=1}^N \omegahat_{k,i}\zeta_i = \E\left(\omegabar_{k,i}\zeta_i\right) + O_p(c_n),
\end{eq}
where $c_n = \atld_N + n^{-1/2}N^{-1}h^{-3}$.
\end{lemma}
\begin{proof}
Consider the decomposition:
\begin{eq*}
N^{-1}\sum_{i=1}^N \omegahat_{k,i}\zeta_i &= \Sscr_{1,k}+\Sscr_{2,k} + \Sscr_{3,k},
\end{eq*}
where:
\begin{eq*}
\Sscr_{1,k} = N^{-1}\sum_{i=1}^N \omegabar_{i,k}\zeta_i, \quad
&\Sscr_{2,k} = N^{-1}\sum_{i=1}^N \left\{\frac{1}{\pihat_k(\bX_i;\bvthetbar)}-\frac{1}{\pi_k(\bX_i;\bvthetbar)}\right\} I(T_i=k)\zeta_i,\\
\text{ and }&\Sscr_{3,k} = N^{-1}\sum_{i=1}^N \left\{ \frac{1}{\pihat(\bX_i;\bvthethat)}- \frac{1}{\pihat_k(\bX_i;\bvthetbar)}\right\}I(T_i=k)\zeta_i.
\end{eq*}

The second term can be bounded:
\begin{eq*}
\abs{\Sscr_{2k}} &\leq \sup_{\bx}\norm{\pihat_k(\bx;\bvthetbar)-\pi_k(\bx;\bvthetbar)}N^{-1}\sum_{i=1}^N \frac{I(T_i = k)\zeta_i}{\pihat_k(\bX_i;\bvthetbar)\pi_k(\bX_i;\bvthetbar)} \\
&= O_p(\atld_N).
\end{eq*}

The third term can be written:
\begin{eq*}
\label{e:S3k}
\Sscr_{3,k} &= N^{-1}\sum_{i=1}^N \ddbalph \frac{I(T_i=k)\zeta_i}{\pihat_k(\bX_i;\balphbar_1,\bbetabar_1)}(\balphhat_1-\balphbar_1) + \ddbbeta \frac{I(T_i = k)\zeta_i}{\pihat_k(\bX_i;\balphbar_1,\bbetabar_1)}(\bbetahat_1-\bbetabar_1) \\
&\qquad + O_p\left(\norm{\balphhat_1-\balphbar_1}^2+\norm{\bbetahat_1-\bbetabar_1}^2 + \norm{\balphhat_1-\balphbar_1}\norm{\bbetahat_1 - \bbetabar_1}\right) \\
&= O_p\left\{(1 + N^{-1/2}h^{-1} + N^{-1}h^{-3})n^{-1/2}\right\}
\end{eq*}
where we use that $\ddbalph1/\pihat_k(\bX_i;\balph_1,\bbeta_1)$ and $\ddbalph1/\pihat_k(\bX_i;\balph_1,\bbeta_1)$ are Lipshitz continuous in $\balph_1$ and $\bbeta_1$ for the first equality and the rate deduced from an analogous term in \cite{cheng2017estimating} for the second equality.  The desired result follows from collecting the dominant rates.
\end{proof}

\begin{lemma}
\label{l:2}
Let $\Ydag_i = \xi_{T_i}(\bV_i;\bgambar,\pi)$ for $i=1,\ldots,N$.  Then:
\begin{eq}
\nhlfNinv \sum_{i=1}^N \omegahat_{k,i}(\Ydag_i - \mubar_k) = O_p(1 + d_n),
\end{eq}
where $d_n = \nu_n^{1/2} N^{\half}h^q + \nu_n^{\half} N^{-\half}h^{-2} + N^{-\half}h^{-1}+N^{-1}h^{-3}$.
\begin{proof}
Consider the decomposition:
\begin{eq*}
\nhlfNinv \sum_{i=1}^N \omegahat_{k,i}(\Ydag_i - \mubar_k) = \Wscrtld_{1,1,k} + \Wscrtld_{2,1,k} + \Wscrtld_{3,1,k},
\end{eq*}
where:
\begin{eq*}
&\Wscrtld_{1,1,k} = \nhlfNinv\sum_{i=1}^N \omegabar_{k,i}(\Ydag_i - \mubar_k) \\
&\Wscrtld_{2,1,k} = \nhlfNinv\sum_{i=1}^N \left\{\frac{1}{\pihat_k(\bX_i;\bvthetbar)} - \frac{1}{\pi_k(\bX_i;\bvthetbar)} \right\}I(T_i =k)(\Ydag_i - \mubar_k) \\
&\Wscrtld_{3,1,k} = \nhlfNinv\sum_{i=1}^N \left\{\frac{1}{\pihat_k(\bX_i;\bvthethat)} - \frac{1}{\pihat_k(\bX_i;\bvthetbar)} \right\}I(T_i =k)(\Ydag_i - \mubar_k).
\end{eq*}

The first term is a scaled sum of iid centered terms so that:
\begin{eq*}
\Wscrtld_{1,1,k} = \nu_n^{\half}N^{-\half}\sum_{i=1}^N \omegabar_{k,i}(\Ydag_i - \mubar_k) = \nu_n^{\half}O_p(1).
\end{eq*}

The V-statistic arguments similar to \cite{cheng2017estimating}, the second term can be written:
\begin{eq*}
\Wscrtld_{2,1,k} &= -\nu_n^{\half}N^{-\half}\sum_{i=1}^N \E(\Ydag_i\mid \bSbar_i,T_i=k) \left\{\frac{I(T_i = k)}{\pi_k(\bX_i;\bvthetbar)}-1\right\} \\
&\qquad + O_p\left\{\nu_n^{\half}(h^q + N^{-\half}h^{-2})\right\} \\
&\qquad + O_p\left\{ \nu_n^{\half}(h^q + N^{\half}h^q + N^{-\half}h^{-2})\right\}+ O_p(\nu_n^{\half}N^{\half}\atld_N^2) \\
&= O_p(\nu_n^{\half}N^{\half}h^q + \nu_n^{\half} N^{-\half}h^{-2}).
\end{eq*}
The final term can be written:
\begin{eq*}
\Wscrtld_{3,1,k} &= \nhlfNinv \sum_{i=1}^N \ddbalph \frac{I(T_i=k)(\Ydag_i-\mubar_k)}{\pihat_k(\bX_i;\balphbar_1,\bbetabar_1)} (\balphhat_1-\balphbar_1)+ \ddbbeta \frac{I(T_i =k)(\Ydag_i -\mubar_k)}{\pihat_k(\bX_i;\balphbar_1,\bbetabar_1)} (\bbetahat_1-\bbetabar_1) \\
&\qquad + O_p\left\{n^{\half}\left( \norm{\balphhat_1-\balphbar_1}^2+\norm{\bbetahat_1-\bbetabar_1}^2+\norm{\balphhat_1-\balphbar_1}\norm{\bbetahat_1-\bbetabar_1} \right)\right\} \\
&= O_p(1+N^{-\half}h^{-1}+N^{-1}h^{-3})O_p(1).
\end{eq*}
where we use that $\ddbalph1/\pihat_k(\bX_i;\balph,\bbeta)$ and $\ddbalph1/\pihat_k(\bX_i;\balph,\bbeta)$ are Lipshitz continuous in $\balph$ and $\bbeta$ for the first equality and used the rate deduced from an analogous term from \cite{cheng2017estimating} for the second equality.  The desired result follows from collecting the rates.
\end{proof}
\end{lemma}

\begin{lemma}
\label{l:npbnd}
Let $\Mscr_{NP} = \{ f_{Y,\bZ}(y,\bz) : \text{ there exists a } \epsilon_{\pi} > 0 \text{ such that } \pi_1(\bx) \in [\epsilon_{\pi},1-\epsilon_{\pi}] \text{ for all } \bx \text{ where } f_{\bX}(\bx)>0\}$ be a nonparametric model for the distribution of $(Y,\bZ)$, where $\bz=(\bv\trans,t)\trans$.  Let $\Mscr_{NP,sub} = \{ f_{Y,\bZ}(y,\bz;\theta):\theta \in \Theta\}$ denote a regular parametric submodel of $\Mscr_{NP}$, where $\theta$ is a finite-dimensional parameter and the true density is at $\theta=\thetastr$.  Let $\Pscr_{NP}$ be the collection of all such regular parametric submodels that satisfy:
\begin{enumerate}
\item $\E_{\theta}[\E_{\thetastr}(Y\mid\bX,T=k)^2]$ is continuous in $\theta$ at $\theta=\thetastr$ for $k=0,1$, where $\E_{\theta}(\cdot)$ and $\E_{\theta}(\cdot \mid \cdot)$ denote expectation and conditional expectation with respect to $f(\cdot;\theta)$ and $f(\cdot\mid\cdot;\theta)$
\item The score at $\thetastr$, satisfies $S_{Y,\bW,T,\bX}(\thetastr)=S_{Y\mid\bW,T,\bX}(\thetastr)+ S_{\bW\mid T,\bX}(\thetastr) + S_{T\mid \bX}(\thetastr) + S_{\bX}(\thetastr)$, where $S_{Y\mid \bW,T,\bX}(\thetastr)$, $S_{\bW\mid T,\bX}(\thetastr)$, $S_{T\mid \bX}(\thetastr)$ and $S_{\bX}(\thetastr)$ denote the scores in implied parametric submodels for the respective conditional and marginal distributions at $\thetastr$.
\item $\ddtheta \E_{\thetastr}\{ \E_{\theta}(Y\mid \bX,T=k)\}\Big|_{\thetastr} = \E_{\thetastr}\{ \ddtheta \E_{\theta}(Y\mid \bX,T=k)\Big|_{\thetastr}\}$ and $\E_{\thetastr}[ \ddtheta \E_{\thetastr}\{\E_{\theta}(Y\mid \bW,\bX,T=k)\mid \bX,T=k\}\Big|_{\thetastr}]= \E_{\thetastr}[ \E_{\thetastr}\{\ddtheta \E_{\theta}(Y\mid \bW,\bX,T=k)\Big|_{\thetastr}\mid \bX,T=k\}]$ for $k=0,1$.
\item $\E_{\theta}\{ \E_{\thetastr}(Y\mid \bW,\bX,T=k)^2\mid \bX,T=k\}$ is continuous in $\theta$ at $\thetastr$ for $k=0,1$.
\item $\E_{\theta}(Y^2\mid \bW,\bX,T=k)$ is continuous in $\theta$ at $\thetastr$ for $k=0,1$.
\end{enumerate}
The efficient influence function for $\Deltabarstr = \E\{ \E(Y\mid \bX,T=1) - \E(Y\mid \bX,T=0)\}$ in $\Mscr_{NP}$ with respect to $\Pscr_{NP}$ is:
\begin{align*}
\Psi_{eff} = \E(Y\mid \bX,T=1)-\E(Y\mid \bX,T=0) + \{ \frac{I(T=1)}{\pi_1(\bX)} - \frac{I(T=0)}{\pi_0(\bX)}\}\{Y - \E(Y\mid \bX,T)\} - \Deltabarstr.
\end{align*}
The semiparametric efficiency bound for $\Deltabarstr$ under $\Mscr_{NP}$ with respect to $\Pscr_{NP}$ is $\E(\Psi_{eff}^2)$.
\end{lemma}
\begin{proof}
The derivation of the semiparametric efficiency bound for $\Deltabarstr$ under $\Mscr_{NP}$ directly follows arguments from the well-known works of \cite{robins1994estimation} and \cite{hahn1998role}. It can be shown that the availability of $\bW$ in our framework does not alter the bound as $\Mscr_{NP}$ is a model for the distribution of data where $Y$ is fully observed. We omit repeating the arguments here for brevity.
\end{proof}

\begin{lemma}
    \label{thm:0}
    Let $\lambda_n = o(n^{-1/2})$ and $\bgamhat$ be the solution to penalized estimating equation for a GLM based on an exponential family with canonical link function:
    \begin{align*}
        \bU_n(\bgam;\lambda_n) = n^{-1}\sum_{i=1}^n \bZpii \left\{ Y_i - g_{\xi}(\bgam\trans\bZpii)\right\} - \lambda_n \bgam_{\circ} = \bzero.
    \end{align*}
    Let the parameter space for $\bgam$ be compact and $\sum_{i=1}^n \bZpii \bZpii\trans$ be full-rank. Then  $\bgamhat$ is a $n^{1/2}$-consistent estimator of $\bgambar$ such that $n^{1/2}(\bgamhat - \bgambar) = O_p(1)$, where $\bgambar$ solves:
    \begin{align*}
        \E\left[ \bZpi \left\{ Y - g_{\xi}(\bgam\trans\bZpi)\right\}\right] = \bzero.
    \end{align*}
    \begin{proof}
        Let the 
        population estimating equation be $\bU(\bgam) = \E\left[ \bZpi \left\{ Y - g_{\xi}(\bgam\trans\bZpi)\right\}\mid L = 1\right] = \E\left[ \bZpi \left\{ Y - g_{\xi}(\bgam\trans\bZpi)\right\}\right]$, using that $L \indep (\bZpi\trans,Y)\trans$ due to random labeling.
        $\bgambar$ is a unique solution to $\bU(\bgam)=\bzero$ as the log-likelihood of a GLM with canonical link is strictly concave.  It also follows from the uniform law of large numbers \cite{newey1994large} that 
        $\sup_{\bgam}\norm{\bU_n(\bgam)-\bU(\bgam)}\overset{p}{\to}0$. Consequently $\bgamhat \overset{p}{\to} \bgambar$ by
        by Z-estimation theory \citep{van2000asymptotic}.
        Next, expanding $\bU_n(\bgamhat)$ around $\bgambar$:
        \begin{align*}
            \bzero = \bU_n(\bgambar) + \frac{\partial}{\partial\bgam}\bU_n(\bgamhat^*) (\bgamhat-\bgambar),
        \end{align*}
        which yields:
        \begin{align*}
            n^{1/2}(\bgamhat-\bgambar) &= -\frac{\partial}{\partial\bgam} \bU_n(\bgamhat^*)^{-1}\left\{n^{-1/2}\sum_{i=1}^n \bZpii \left\{ Y_i - g_{\xi}(\bgambar\trans\bZpii)\right\} - n^{1/2}\lambda_n \bgambar_{\circ} \right\},
        \end{align*}
        where $\bU_n(\bgamhat^*)$ denotes a Jacobian matrix with each row evaluated at a different intermediate point $\bgamhat^*$ between $\bgamhat$ and $\bgambar$.  Another application of the uniform law of large numbers yields that $\sup_{\bgam}\norm{\frac{\partial}{\partial\bgam}\bU_n(\bgam) - \frac{\partial}{\partial\bgam}\bU(\bgam)}\overset{p}{\to}0$.
        This, along with the fact that $\bgamhat\overset{p}{\to}\bgambar$ yields that $\frac{\partial}{\partial\bgam}\bU_n(\bgamhat^*)\overset{p}{\to}U(\bgambar)$. The penalty terms can be seen to be asymptotically negligible when $\lambda_n = o(n^{-1/2})$.  By applications of Central Limit Theorem and Slutsky's Theorem, we have that $n^{1/2}(\bgamhat - \bgambar) = O_p(1)$.
    \end{proof}
\end{lemma}

\section*{Web Appendix B: Consistency and Asymptotic Linearity of $\Delthat$}
\begin{theorem}
\label{thm:1}
Under the identification assumptions from (2)-(4) of the main text,
given a bandwidth of $h = O(N^{-\alpha})$ for $\frac{1-\beta}{2q}<\alpha <  \min(\frac{\beta}{2},\frac{1}{4})$ and $n = O(N^{1-\beta})$ with $\frac{1}{q+1}<\beta <1$, $\Delthat - \Delta = O_p(n^{-\half})$ when either $\pi_1(\bx;\balph)$ or $\mu_k(\bx;\bbeta)$ is correctly specified.
\end{theorem} 
\begin{proof}
We first show that $\Delthat - \Deltbar=O_p(n^{-\half})$ where $\Deltbar = \mubar_1 - \mubar_0$.  If this can be shown, the limiting estimate is:
\begin{eq*}
\Deltbar &= \E\left\{ \frac{I(T=1)Y}{\pi_1(\bX;\bvthetbar)}-\frac{I(T=0)Y}{\pi_0(\bX;\bvthetbar)}\right\} = \Delta,
\end{eq*}
where the first equality follows from the argument in the main text and the second equality holds when either $\pi_1(\bx;\balph)$ or $\mu_k(\bx;\bbeta)$ are correctly specified \citep{cheng2017estimating}.  It suffices to show that $\muhat_k - \mubar_k = O_p(n^{-\half})$ , for $k=0,1$.  First note that the normalizing constant can effectively be ignored.  By application of Lemma \ref{l:1} with $\zeta_i = 1$, the normalizing constant is:
\begin{eq}
\label{e:nrmrte}
\ntotinv\sum_{i=1}^{\ntot} \omegahat_{k,i} = 1 +O_p(c_n).
\end{eq}

We can now write the standardized mean for the $k$-th group as:
\begin{eq}
\label{e:nrmstdmean}
n^{\half}(\muhat_k -\mubar_k) &= \nhlfNinv\sum_{i=1}^{\ntot}\omegahat_{k,i} (\Yhatdag_i-\mubar_k) + \left[ \left\{ \ntotinv\sum_{i=1}^{\ntot}\omegahat_{k,i}\right\}^{-1} - 1\right]\nhlfNinv\sum_{i=1}^{\ntot}\omegahat_{k,i} (\Yhatdag_i-\mubar_k) \nonumber\\
&= \nhlfNinv\sum_{i=1}^{\ntot}\omegahat_{k,i} (\Yhatdag_i - \mubar_k)+ O_p(c_n),
\end{eq}
where the last equality follows provided that the main term is $O_p(1)$.  Denote: $$\Wscrtld_k= \nhlfNinv\sum_{i=1}^{\ntot}\omegahat_{k,i} (\Yhatdag_i - \mubar_k).$$ This can be decomposed as $\Wscrtld_k = \Wscrtld_{1,k} + \Wscrtld_{2,k}$, where:
\begin{eq}
\label{e:meandcmp}
\Wscrtld_{1,k} = \nhlfNinv\sum_{i=1}^{\ntot}\omegahat_{k,i} (\Ydag_i - \mubar_k) \text{ and }\Wscrtld_{2,k} =  \nhlfNinv\sum_{i=1}^{\ntot}\omegahat_{k,i}(\Yhatdag_i - \Ydag_i).
\end{eq}

First, focusing on $\Wscrtld_{2,k}$, we expand $\Yhatdag_i = \xi_{T_i}(\bV_i;\bgamhat,\pihat)$ around $\Ydag_i = \xi_{T_i}(\bV_i;\bgambar,\pi)$:
\begin{eq}
\label{e:impdcmp}
\Wscrtld_{2,k} &= \Wscrtld_{2,k}^{\pi} + \Wscrtld_{2,k}^{\bgam} + O_p\left\{\sup_{\bx}\norm{\pihat(\bx;\bvthethat)-\pi(\bx;\bvthetbar)}^2\right\}+O_p(\norm{\bgamhat - \bgambar}^2),
\end{eq}
where
\begin{eq}
\label{e:residcmp}
\Wscrtld_{2,k}^{\pi} &= \nhlfNinv\sum_{i=1}^N \omegahat_{k,i}\ddpi \xi_{T_i}(\bV_i;\bgambar,\pi)\left\{ \pihat_k(\bX_i;\bvthethat)-\pi_k(\bX_i;\bvthetbar)\right\}
\end{eq}
accounts for estimating the DiPS in the imputation with 
$\ddpi \xi_{T_i}(\bV_i;\bgambar,\pi) = \gdot_{\xi}(\bgambar\trans\bZ_{\pi,i}) [\bar{\gamma}_3\{- \frac{I(T_i=1)}{\pi_{k}(\bX_i;\bvthetbar)^2}-\frac{I(T_i=0)}{(1-\pi_{k}(\bX_i;\bvthetbar))^2}\}]$ and $\dot{g}_{\xi}(u)=dg_{\xi}(u)/du$ and
\begin{eq}
\Wscrtld_{2,k}^{\bgam} 
&= \nhlfNinv\sum_{i=1}^N \omegahat_{k,i} \ddbgam\xi_{T_i}(\bV_i;\bgambar,\pi) (\bgamhat - \bgambar) + O_p(\sup_{\bx}\norm{\pihat(\bx;\bvthethat)-\pi(\bx;\bvthetbar)}\norm{\bgamhat - \bgambar})
\end{eq}
accounts for estimating $\bgam$ in the imputation. 
We can further decompose $\Wscrtld_{2,k}^{\pi} = \Wscrtld_{2,k}^{pa,\pi} + \Wscrtld_{2,k}^{np,\pi}$, where:
\begin{eq}
\label{e:W2pidcmp}
&\Wscrtld_{2,k}^{pa,\pi} = \nhlfNinv\sum_{i=1}^N \omegahat_{k,i}\ddpi \xi_{T_i}(\bV_i;\bgambar,\pi) \left\{ \pihat_k(\bX_i;\bvthethat)-\pihat_k(\bX_i;\bvthetbar)\right\} \nonumber\\
&\Wscrtld_{2,k}^{np,\pi} = \nhlfNinv\sum_{i=1}^N \omegahat_{k,i}\ddpi \xi_{T_i}(\bV_i;\bgambar,\pi) \left\{ \pihat_k(\bX_i;\balphbar_1,\bbetabar_1)-\pi_k(\bX_i;\bvthetbar)\right\}.
\end{eq}

The $\Wscrtld_{2,k}^{pa,\pi}$ term accounts for the parametric estimation of $\balph_1$ and $\bbeta_1$ in DiPS and is:
\begin{eq}
\label{e:W2pipa}
&\Wscrtld_{2,k}^{pa,\pi} =\nhlfNinv\sum_{i=1}^N \omegahat_{k,i}\ddpi \xi_{T_i}(\bV_i;\bgambar,\pi) \bigg\{\ddbalph\pihat_k(\bX_i;\balphbar_1,\bbetabar_1)(\balphhat - \balphbar) \nonumber\\
&\qquad+\ddbbeta\pihat_k(\bX_i;\balphbar_1,\bbetabar_1)(\bbetahat - \bbetabar) + O_p(\norm{\balphhat-\balphbar}^2+\norm{\bbetahat-\bbetabar}^2+\norm{\balphhat-\balphbar}\norm{\bbetahat-\bbetabar})\bigg\} \nonumber\\
&= \ntotinv\sum_{i=1}^N \omegahat_{k,i}\ddpi\xi_{T_i}(\bV_i;\bgambar,\pi) \left\{ \ddbalph\pi_k(\bX_i;\balphbar_1,\bbetabar_1)n^{\half}(\balphhat - \balphbar) + \ddbbeta\pi_k(\bX_i;\balphbar_1,\bbetabar_1)n^{\half}(\bbetahat - \bbetabar)\right\} \nonumber\\
&\qquad+ O_p(1+c_n)O_p(\btld_N)O_p(\nu_n^{\half})+O_p(1+c_n)O_p(\btld_N)O_p(1) \nonumber\\
&= O_p(1+c_n) + O_p(\btld_N),
\end{eq}
where the first equality uses the Lipshitz continuity of $\dduK(\cdot)$, the second equality applies Lemma \ref{l:0} and Lemma \ref{l:1} taking $\zeta_i=\ddpi\xi_{T_i}(\bV_i;\bgambar,\pi)$, and the last equality applies lemma \ref{l:1} again taking $\zeta_i=\ddpi\xi_{T_i}(\bV_i;\bgambar,\pi)\ddbalph\pi(\bX_i;\bvthetbar)$ as well as $\zeta_i=\ddpi\xi_{T_i}(\bV_i;\bgambar,\pi)\ddbbeta\pi(\bX_i;\bvthetbar)$.

The $\Wscrtld_{2,k}^{np,\pi}$ term accounts for the nonparametric smoothing in DiPS and can be bounded:
\begin{eq}
\label{e:W2pinp}
\Wscrtld_{2,k}^{np,\pi} &\leq n^{\half}\sup_{\bx}\norm{\pihat_k(\bx;\bvthetbar)-\pi_k(\bx;\bvthetbar)}\ntotinv\sum_{i=1}^N \omegahat_{k,i}\ddpi\xi_{T_i}(\bV_i;\bgambar,\pi) \\
&= O_p(n^{\half}\atld_N)O_p(1+c_n)= O_p(n^{\half}\atld_N).
\end{eq}

Returning to \eqref{e:residcmp}, the following term accounts for the parametric estimation of $\bgam$:
\begin{eq}
\label{e:W2ga}
\Wscrtld_{2,k}^{\bgam} &= \ntotinv \sum_{i=1}^N \omegahat_{k,i}\ddbgam\xi_{T_i}(\bV_i;\bgambar,\pi) n^{\half}(\bgamhat - \bgambar)+ O_p(1+c_n)O_p(a_n)O_p(1) \nonumber \\
&= O_p(1+c_n)O_p(1) + O_p(a_n) = O_p(1+c_n+a_n),
\end{eq}
where the first equality applies Lemma \ref{l:1} taking $\zeta_i = \ddbgam\xi_{T_i}(\bV_i;\bgambar,\pi)\ddpi\xi_{T_i}(\bV_i;\bgambar,\pi)$ as well as Lemma \ref{l:0}, and the second equality follows from application of Lemma \ref{l:1} again taking $\zeta_i = \ddbgam\xi_{T_i}(\bV_i;\bgambar,\pi)$.  Finally, collecting all the terms, we find:
\begin{eq}
n^{1/2}(\muhat_k - \mubar_k) 
&= \Wscrtld_{1,k} + \Wscrtld_{2,k}^{pa,\pi} +\Wscrtld_{2,k}^{np,\pi}+ \Wscrtld_{2,k}^{\bgam} + O_p( a_n^2 +n^{-1}) + O_p(c_n) \nonumber\\
&= O_p(1+d_n) + O_p(1+c_n + \btld_N) + O_p(n^{\half}\atld_N) + O_p(1+c_n+a_n) \nonumber\\
&\qquad+ O_p( a_n^2 +n^{-1}) + O_p(c_n) \nonumber\\
&= O_p(1),
\end{eq}
where the second to last equality applies Lemma \ref{l:2} and the last equality follows when $h=O(N^{-\alpha})$ for $\frac{1-\beta}{2q}<\alpha <  \min(\frac{\beta}{2},\frac{1}{4})$ and $n = O(N^{1-\beta})$ for $\frac{1}{q+1}<\beta <1$. This shows that $\muhat_k - \mubar_k = O_p(n^{-\half})$.
\end{proof}

\begin{theorem}
\label{thm:2}
Let $\Wscrhat_k=n^{\half}(\muhat_k - \mubar_k) $ for $k=0,1$ so that $n^{\half}(\Delthat - \Deltbar)= \Wscrhat_1 - \Wscrhat_0$. Given a bandwidth of $h = O(N^{-\alpha})$ for $\frac{1-\beta}{2q}< \alpha <  \min(\frac{\beta}{4},\frac{1}{4})$ and $n = O(N^{1-\beta})$ with $\frac{1}{q+1}<\beta <1$, then $\Wscrhat_k$ has the influence function expansion of the form:
\begin{eq}
\Wscrhat_k &= n^{-1/2}\sum_{i=1}^n (\bv_{\bbeta_1,k}\trans+\bu_{pa,\pi,k}\trans)\bvphi_{\bbeta_1,i} + \bu_{\bgam,k}\trans\bvphi_{\bgam,i} + o_p(1),
\end{eq}
where $\bv_{\bbeta_1,k}=\bzero$ when $\pi_1(\bx;\balph)$ is correctly specified and $\bu_{pa,\pi,k}=\bzero$ when either $\pi_1(\bx;\balph)$ or $\xi_k(\bv;\bgam,\pi)$ without the utility covariate is correctly specified.
\end{theorem}
\begin{proof}
As in \eqref{e:nrmstdmean} and \eqref{e:meandcmp} the standardized mean for the $k$-th group can be written:
\begin{eq}
\label{e:maindcmp}
\Wscrhat_k = \Wscrtld_k + O_p(c_n) = \Wscrtld_{1,k} +\Wscrtld_{2,k} + O_p(c_n) ,
\end{eq}
where the first equality follows provided $\Wscrtld_k = O_p(1)$.  The first term can be written:
\begin{eq*}
&\Wscrtld_{1,k} = \nu_n^{\half} N^{-\half}\sum_{i=1}^N \omegahat_{k,i}(\Ydag_i-\mubar_k) \\
&= \nu_n^{\half}\Wscrtld_{1,1,k}+\nu_n^{\half}\Wscrtld_{2,1,k}^{ct} + O_p\left\{\nu_n^{\half}(h^q + N^{1/2}h^q + N^{-1/2}h^{-2} + N^{\half}\atld_N^2)\right\} \\
&+ \left\{\bv_{\bbeta_1,k}\trans+O_p\left(N^{-\half}h^{-1}+N^{-1}h^{-3}\right)\right\}n^{\half}(\bbetahat_1-\bbetabar_1)  + O_p\left\{\nu_n^{\half}(1+N^{-\half}h^{-1}+N^{-1}h^{-3})\right\},
\end{eq*}
where:
\begin{eq*}
&\Wscrtld_{1,1,k} = N^{-1/2}\sum_{i=1}^N \omegabar_{k,i}(\Ydag_i-\mubar_k) \\
&\Wscrtld_{2,1,k}^{ct} = - N^{-1/2}\sum_{i=1}^N \E(\Ydag \mid \bSbar,T=k)\left( \omegabar_{k,i}-1\right) + O_p(h^q + N^{-1/2}h^{-2}) \\
&\bv_{\bbeta_1,k} = \E\left\{ \dot{K}_h\left(\frac{\bSbar_j - \bSbar_i}{h}\right)\trans\left( 1-\omegabar_{k,i}\right)\frac{I(T_i=k)(\Ydag_i-\mubar_k)}{l_k(\bX_i;\bvthetbar)}(\bX_j^{\dagger} - \bX_i^{\dagger})\trans\right\},
\end{eq*}
with $l_k(\bx;\bvthetbar) = \pi_k(\bx;\bvthetbar)f(\bx;\bvthetbar)$ and $f(\bx;\bvthetbar)$ being the joint density of $\bSbar$, and $\bx^{\dagger}=(\bx,\bzero)$ for any vector $\bx$.  Here $\bv_{\bbeta_1,k}=O_p(1)$ in general and is $\bzero$  when $\pi(\bX;\balph)$ is correctly specified \citep{cheng2017estimating}.  As in \eqref{e:impdcmp} and \eqref{e:W2pidcmp} of Theorem \ref{thm:1}, the second term from \eqref{e:maindcmp} can be written:
\begin{eq*}
\Wscrtld_{2,k} = \Wscrtld_{2,k}^{pa,\pi} +\Wscrtld_{2,k}^{np,\pi}+ \Wscrtld_{2,k}^{\bgam} + O_p(a_n^2 + n^{-1}).
\end{eq*}

Continuing the expansion of $\Wscrtld_{2,k}^{pa,\pi}$ from \eqref{e:W2pipa}:
\begin{eq*}
\Wscrtld_{2,k}^{pa,\pi} &= N^{-1}\sum_{i=1}^N\omegahat_{k,i}\ddpi\xi_{T_i}(\bV_i;\bgambar,\pi)\ddbbeta\pi_k(\bX_i;\bvthetbar)n^{\half}(\bbetahat_1 - \bbetabar_1) + O_p\left\{ (1+c_n)(\nu_n^{\half} + \btld_N)\right\}\\
&= \E\left\{ \omegabar_{k,i}\ddpi\xi_{T_i}(\bV_i;\bgambar,\pi)\ddbbeta\pi_k(\bX_i;\bvthetbar)\right\}n^{\half}(\bbetahat_1 - \bbetabar_1) + O_p\left\{ c_n + (1+c_n)(\nu_n^{\half} + \btld_N)\right\} \\
&= n^{-\half}\sum_{i=1}^n \bu_{pa,\pi,k}\trans \bvphi_{\bbeta_1,i} + o_p(1) + O_p(c_n+\btld_N),
\end{eq*}
by repeated application of Lemma \ref{l:1}, where:
\begin{eq}
\bu_{pa,\pi,k}\trans = \E\left\{ \omegabar_{k,i}\ddpi\xi_{T_i}(\bV_i;\bgambar,\pi)\ddbbeta\pi_k(\bX_i;\bvthetbar)\right\}
\end{eq}
is a constant that is $\bzero\trans$ when either $\pi_k(\bx;\balph)$ or $\xi_k(\bv;\bgam,\pi)$ without the utility covariate is correctly specified and $\bvphi_{\bbeta_1,i}$ is the influence function for $\bbetahat$.  For $\Wscrtld_{2,k}^{np,\pi}$ from \eqref{e:W2pinp} we have that $\Wscrtld_{2,k}^{np,\pi} = O_p(n^{1/2}\atld_N)$.
For $\Wscrtld_{2,k}^{\bgam}$, continuing from \eqref{e:W2ga}:
\begin{eq*}
\Wscrtld_{2,k}^{\bgam} &= N^{-1}\sum_{i=1}^N \omegahat_{k,i} \ddbgam \xi_{T_i}(\bV_i;\bgambar,\pi)n^{1/2}(\bgamhat-\bgambar)+ O_p(a_n)\\
&= \E\left\{ \omegabar_{k,i} \ddbgam \xi_{T_i}(\bV_i;\bgambar,\pi)\right\} n^{\half}(\bgamhat - \bgambar) + O_p(c_n)O_p(1) + O_p(a_n) \\
&= n^{-1/2}\sum_{i=1}^n \bu_{\bgam,k}\trans \bvphi_{\bgam,i} + o_p(1) + O_p(c_n+a_n),
\end{eq*}
where:
\begin{eq*}
\bu_{\bgam,k}\trans = \E\left\{ \omegabar_{k,i} \ddbgam \xi_{T_i}(\bV_i;\bgambar,\pi)\right\}
\end{eq*}
is some constant and $\bvphi_i^{\bgam}$ is the influence function for $\bgamhat$.

Collecting the results from above, we find:
\begin{eq*}
\Wscrhat_k &= n^{-1/2}\sum_{i=1}^n (\bv_{\bbeta_1,k}\trans+\bu_{pa,\pi,k}\trans)\bvphi_{\bbeta_1,i} + \bu_{\bgam,k}\trans\bvphi_{\bgam,i} + O_p\left\{\nu_n^{1/2}(N^{1/2}h^q +N^{-1/2}h^{-2})\right\} \\
&\qquad+ O_p(N^{-1/2}h^{-1}+N^{-1}h^{-3} + \btld_N + n^{\half}\atld_N + c_n + a_n) \\
&= n^{-1/2}\sum_{i=1}^n (\bv_{\bbeta_1,k}\trans+\bu_{pa,\pi,k}\trans)\bvphi_{\bbeta_1,i} + \bu_{\bgam,k}\trans\bvphi_{\bgam,i} + O_p\left\{\nu_n^{1/2}N^{1/2}h^q + \frac{(logN)^{1/2}}{N^{1/2}h^2}+\nu_n^{1/2}\frac{(logN)^{1/2}}{h}\right\}.
\end{eq*}
The error terms are $o_p(1)$ when $h=O(N^{-\alpha})$ for $\frac{1-\beta}{2q}<\alpha<\min(\frac{\beta}{2},\frac{1}{4})$ and $n=O(N^{1-\beta})$ for $\frac{1}{q+1}<\beta <1$.
\end{proof}

\section*{Web Appendix C: Semiparametric Efficiency}

\begin{theorem}
\label{thm:semipar}
Let $\Mscr_{SS}=\{ f_{Y,\bZ}(y,\bz) = f_{Y\mid \bZ}(y\mid \bz)\fstr_{\bZ}(\bz): \fstr_{\bZ}(\bz)$ is a known density such that there exists a $\epsilon_{\pi} > 0 \text{ where } \pi^*_1(\bx) \in [\epsilon_{\pi},1-\epsilon_{\pi}] \text{ for all } \bx \text{ with } f_{\bX}^*(\bx) >0 \} $ be an ideal semiparametric semi-supervised model where the distribution of $\bZ = (\bV\trans,T)\trans$ is known, with $\bz = (\bv\trans,t)\trans$ and $\pi_1^*(\bx)$ and $f_{\bX}^*(\bx)$ being the implied PS and density of $\bX$ under $f_{\bZ}^*(\bz)$.  Let $\Mscr_{SS,sub}=\{ f_{Y\mid \bZ}(y,\bz;\theta)f_{\bZ}^*(\bz) : \theta \in \Theta\}$ denote a regular parametric submodel of $\Mscr_{SS}$, where $\theta$ is a finite-dimensional parameter, and the true density is at $\theta=\thetastr$.  Let $\Pscr_{SS} \subseteq \Pscr_{NP}$ be the sub-collection of all such regular parametric submodels among $\Pscr_{NP}$.
The efficient influence function for $\Deltabarstr$ in $\Mscr_{SS}$ with respect to $\Pscr_{SS}$ is:
\begin{eq}
\varphi_{eff} = \left\{ \frac{I(T=1)}{\pi_1(\bX)} - \frac{I(T=0)}{\pi_0(\bX)}\right\}\left\{ Y- \xi_T(\bV)\right\},
\end{eq}
and the semiparametric efficiency bound for $\Deltabarstr$ under $\Mscr_{SS}$ with respect to $\Pscr_{SS}$ is $E(\varphi_{eff}^2)$.  Furthermore, the efficiency bound under $\Mscr_{SS}$ is lower than or equal to the efficiency bound under the fully nonparametric model $\Mscr_{NP}$ where the distribution of $\bZ$ is unknown.  That is, $\E(\varphi_{eff}^2)\leq \E(\Psi_{eff}^2)$.
\end{theorem}
\begin{proof}

Let $\Lscralt_2^0$ denote the Hilbert space of mean $0$ square-integrable functions of $(Y,\bZ\trans)\trans$ at the true distribution, with inner product of $v_1,v_2 \in \Lscralt_2^0$ defined by $\langle v_1(Y,\bZ),v_2(Y,\bZ)\rangle = \E \left\{ v_1(Y,\bZ)v_2(Y,\bZ)\right\}$.  We first show that the tangent space of $\Mscr_{SS}$ with respect to $\Pscr_{SS}$ at the true distribution is the closure of $\{ s(Y,\bZ) \in \Lscralt_2^0: \E\{ s(Y,\bZ) \mid  \bZ\}=0\}$, denoted by $\Lambda_{SS}$.  Let $s(Y,\bZ)$ be any bounded element belonging to $\Lambda_{SS}$.  Consider the parametric submodel given by $\Mscr_{SS,tlt} = \{ f_{Y,\bZ}(y,\bz;\theta)=f_{Y\mid \bZ}(y\mid \bz;\theta)f_{\bz}^*(\bz):\theta \in (-\varepsilon,\varepsilon)\}$ for some sufficiently small $\varepsilon > 0$, where:
\begin{eq*}
f_{Y\mid \bZ}(y\mid \bz;\theta) = f_{Y\mid \bZ}^*(y\mid \bz)\{ 1+ \theta s(y,\bz)\}
\end{eq*}
with $f_{Y\mid \bZ}^*(y\mid \bz)$ being the true density.  The true density is thus at $\theta=\thetastr=0$.  It can be shown that $\Mscr_{SS,tlt}$ and the implied conditional and marginal submodels have proper densities, are regular, and the respective score can be written as the derivative of the log density with respect to $\theta$.  It can also be shown through calculations similar to those in analogous arguments for the derivation of Lemma \ref{l:npbnd} that $\Mscr_{SS,tlt}$ belongs in $\Pscr_{SS} \subseteq \Pscr_{NP}$.

The score for $\Mscr_{SS,tlt}$ at $\theta=\thetastr=0$ is $S_{Y,\bZ}(\thetastr) = s(Y,\bZ)$, so any bounded element in $\Lambda_{SS}$ belongs in the tangent space of $\Mscr_{SS}$ with respect to $\Pscr_{SS}$ at the true distribution.  Since the bounded elements are dense in $\Lambda_{SS}$ and the tangent space is closed, any element $r(Y,\bW,T,\bX) \in \Lambda_{SS}$ also belongs in the tangent space.  Any element of the tangent space at the true distribution also belongs in $\Lambda_{SS}$ by the regularity of the parametric submodels and properties of scores.  This verifies that the tangent space of $\Mscr_{SS}$ with respect to $\Pscr_{SS}$ at the true distribution is $\Lambda_{SS}$.

We next show that $\Psi_{eff}$ is one influence function for $\Deltabarstr$ in $\Mscr_{SS}$ at the true distribution with respect to $\Pscr_{SS}$.  Recall from Lemma \ref{l:npbnd} that $\Psi_{eff}$ is the unique influence function for $\Deltabarstr$ in $\Mscr_{NP}$ with respect to $\Pscr_{NP}$.  This means that under any regular parametric submodel $\Mscr_{NP,sub}$ belonging to $\Pscr_{NP}$, $\Psi_{eff}$ satisfies:
\begin{eq*}
\ddtheta \Deltabarstr(\theta)\Big|_{\thetastr}  = \E_{\thetastr}\{\Psi_{eff}S_{Y,\bW,T,\bX}(\thetastr)\}.
\end{eq*}
Now since $\Pscr_{SS}\subseteq \Pscr_{NP}$, pathwise differentiability of $\Deltabarstr(\theta)$ at $\theta=\thetastr$ also holds, in particular, under any regular parametric submodel in $\Pscr_{SS}$, with $\Psi_{eff}$ being one influence function. 

Finally, to obtain the efficient influence function for $\Deltabarstr$ in $\Mscr_{SS}$ with respect to $\Pscr_{SS}$ at the true distribution, we identify the orthogonal projection of  $\Psi_{eff}$ onto $\Lambda_{SS}$.  It can be verified that this projection is $\Pi(\Psi_{eff}\mid \Lambda_{SS}) = \Psi_{eff} - \E(\Psi_{eff}\mid \bZ)$.  The efficient influence function in $\Mscr_{SS}$ is thus:
\begin{eq*}
\varphi_{eff}&=\Pi(\Psi_{eff}\mid \Lambda_{SS})=\Psi_{eff}-\E(\Psi_{eff}\mid \bZ) \\
&= (\E(Y\mid \bX,T=1) + [ \frac{I(T=1)}{\pi_1(\bX)}\{Y-\E(Y\mid \bX,T=1)\}]) \\
&\qquad - (\E(Y\mid \bX,T=0) + [ \frac{I(T=0)}{\pi_0(\bX)}\{Y-\E(Y\mid \bX,T=0)\}]) - \Deltabarstr \\
&\qquad -(\E(Y\mid \bX,T=1) + [ \frac{I(T=1)}{\pi_1(\bX)}\{\E(Y\mid \bZ)-\E(Y\mid \bX,T=1)\}]) \\
&\qquad + (\E(Y\mid \bX,T=0) + [ \frac{I(T=0)}{\pi_0(\bX)}\{\E(Y\mid \bZ)-\E(Y\mid \bX,T=0)\}]) +\Deltabarstr \\
&= \{ \frac{I(T=1)}{\pi_1(\bX)} - \frac{I(T=0)}{\pi_0(\bX)}\} \{Y - \E(Y\mid \bZ)\}.
\end{eq*}
By the Pythagorean theorem, we can verify:
\begin{eq*}
\E(\Psi_{eff}^2) &= \norm{ \Psi_{eff}}_{\Lscralt_2^0}^2 = \norm{\Pi(\Psi_{eff}\mid \Lambda_{SS})}^2_{\Lscralt_2^0} +\norm{\Psi_{eff}-\Pi(\Psi_{eff}\mid \Lambda_{SS})}^2_{\Lscralt_2^0} \\
&\geq \norm{\Pi(\Psi_{eff}\mid \Lambda_{SS})}^2_{\Lscralt_2^0} \\
&=\E(\varphi_{eff}^2).
\end{eq*}
\end{proof}

\begin{corollary}
Given a bandwidth of $h = O(N^{-\alpha})$ for $\frac{1-\beta}{2q}<\alpha < \min(\frac{\beta}{4},\frac{1}{4})$ and $n = O(N^{1-\beta})$ for $\frac{1}{q+1} < \beta < 1$, when $\pi_1(\bx;\balph)$ and $\xi_k(\bv;\bgam,\pi)$ are correctly specified, then:
\begin{eq}\label{e:aipweff}
n^{1/2}(\Delthat - \Deltabarstr) = n^{-\half}\sum_{i=1}^n U_i \{Y_i - \xi_{T_i}(\bV_i)\} + o_p(1).
\end{eq}
That is, $\Delthat$ achieves the semiparametric efficiency bound in the ideal SS semiparametric model where the distribution of $\bZ$ is known.
\begin{proof}
From Theorem \ref{thm:2}, given an appropriate bandwidth and order of labels, when $\pi_1(\bx;\balph)$ is correctly specified:
\begin{eq*}
n^{\half}(\Delthat - \Deltabarstr) = n^{-\half}\sum_{i=1}^n (\bu_{\bgam,1}\trans-\bu_{\bgam,0}\trans)\bvphi_{\bgam,i}  + o_p(1), 
\end{eq*}
where:
\begin{eq*}
\bu_{\bgam,k}\trans\bvphi_{\bgam,i} &= \E\left\{ \omega_{k,i} \ddbgamT \xi_{T_i}(\bV_i;\bgam,\pi)\Big|_{\bgam=\bgambar}\right\}\left\{\ddbgamT \E\bZ_{\pi,i}\xi_{T_i}(\bV_i;\bgam,\pi)\Big|_{\bgam=\bgambar}\right\}^{-1} \bZ_{\pi,i}\left\{Y_i - \xi_{T_i}(\bV_i;\bgambar,\pi)\right\} \\
&= \E\left\{ \omega_{k,i} \bZ_{\pi,i}\trans\dot{g}_{\xi}(\bgambar\trans\bZ_{\pi,i})\right\}\E\left\{\bZ_{\pi,i}\bZ_{\pi,i}\trans \dot{g}_{\xi}(\bgambar\trans\bZ_{\pi,i})\right\}^{-1}\bZ_{\pi,i}\left\{Y_i - \xi_{T_i}(\bV_i;\bgambar,\pi)\right\}.
\end{eq*}
The first and second equalities assume the usual regularity conditions to obtain the influence function of an estimator that is the solution of an estimating equation and exchange order of differentiation and integration.  The influence function for $\Delthat$ can then be written as:
\begin{eq*}
&(\bu_{\bgam,1}\trans-\bu_{\bgam,0}\trans)\bvphi_{\bgam,i} = \E\left\{ U_{\pi,i} \bZ_{\pi,i}\trans \dot{g}(\bgambar\trans\bZ_{\pi,i})\right\}\E\left\{\bZ_{\pi,i}\bZ_{\pi,i}\trans \dot{g}(\bgambar\trans\bZ_{\pi,i})\right\}^{-1}\bZ_{\pi,i}\left\{Y_i - \xi_{T_i}(\bV_i;\bgambar,\pi)\right\}.
\end{eq*}
The terms involving $\bZ_{\pi,i}$ is a population weighted least square projection of $U_i$ onto $\bZ_i$, weighted by $\dot{g}(\bgambar\trans\bZ_{\pi,i})$.  But since $\bZ_{\pi,i}$ includes $U_{\pi,i}$, the influence function simplifies:
\begin{eq*}
(\bu_{\bgam,1}\trans-\bu_{\bgam,0}\trans)\bvphi_{\bgam,i} &= U_{\pi,i} \left\{Y_i - \xi_{T_i}(\bV_i;\bgambar,\pi)\right\} =U_{\pi,i} \left\{Y_i - \xi_{T_i}(\bV_i)\right\},
\end{eq*}
where the second equality follows when $\xi_k(\bv;\bgam,\pi)$ is correctly specified.
\end{proof}
\end{corollary}

\begin{corollary}
Let $\SSPrePost$ be denoted by $\DelthatAIPW = \muhatoneAIPW-\muhatzerAIPW$. Suppose that $\DelthatAIPW$ uses parametric 
working models $\pi(\bx;\balph)$, $\mu_k(\bx;\bbeta)$ and $\xi_k(\bv;\bgam)$ for estimating $\pi(\bx)$, $\mu_k(\bx)$, and
$\xi_k(\bv)$, respectively.  Let $\balphhat,\bbetahat,\bgamhat$ be the respective maximum likelihood estimators and $\balphbar,\bbetabar,\bgambar$
be their probability limits, regardless of the adequacy of each working model. When $\pi(\bx;\balph)$ and $\xi_k(\bv;\bgam)$ are correctly specified so that 
$\pi(\bx;\balphbar)=\pi(\bx)$ and $\xi_k(\bv;\bgambar) = \xi_k(\bv)$, then:
\begin{eq}
    n^{\half}(\DelthatAIPW - \Deltabarstr) = n^{-\half}\sum_{i=1}^n U_i\{ Y_i - \xi_{T_i}(\bV_i)\} + o_p(1).
\end{eq}
That is, $\DelthatAIPW$ also achieves the semiparametric efficiency bound in the ideal SS semiparametric model where the distribution of $\bZ$ is known.
\end{corollary}
\begin{proof}
    Let $\mubar_k^* = \E\{ \mu_k(\bX)\}$ with $\Deltabarstr=\mubarstr_1-\mubarstr_0$. The centered and scaled mean for 
    the $k$-th group can be decomposed as:
    \begin{eq*} 
        \sqrtn(\muhatkAIPW-\mubarstr_k) = \Tscr_{1,k} + \Tscr_{2,k} + \Tscr_{3,k},
    \end{eq*}
    where:
    \begin{eq*}
        &\Tscr_{1,k} = \frac{\sqrt{n}}{N}\sum_{i=1}^N \frac{L_iI(T_i=k)}{\nu_n\pi_k(\bX_i;\balphhat)}\left\{Y_i - \xi_k(\bV_i;\bgamhat)\right\} \\
        &\Tscr_{2,k} = -\frac{\sqrt{n}}{N} \sum_{i=1}^N \left\{\frac{I(T_i =k )}{\pi_k(\bX_i;\balphhat)}-1\right\}\mu_k(\bX_i;\bbetahat) \\
        &\Tscr_{3,k} = \frac{\sqrt{n}}{N}\sum_{i=1}^N \left\{\frac{I(T_i = k)}{\pi_k(\bX_i;\balphhat)}\xi_k(\bV_i;\bgamhat) - \mubarstr_k\right\}.
    \end{eq*}
    The first term can be decomposed as:
    \begin{eq*}
        \Tscr_{1,k} &= \frac{\sqrt{n}}{N}\sum_{i=1}^N \frac{L_iI(T_i=k)}{\nu_n \pi_k(\bX_i)}\left\{ Y_i - \xi_k(\bV_i)\right\}  - \frac{\sqrt{n}}{N}\sum_{i=1}^N \frac{L_iI(T_i=k)}{\nu_n \pi_k(\bX_i)}\left\{ \xi(\bV_i;\bgamhat) - \xi_k(\bV_i)\right\}  + \\
        &\qquad \frac{\sqrt{n}}{N}\sum_{i=1}^N \frac{L_iI(T_i=k)}{\nu_n} \left\{ \frac{1}{\pi_k(\bX_i;\balphhat)}- \frac{1}{\pi_k(\bX_i)}\right\}\left\{ Y_i - \xi(\bV_i;\bgamhat)\right\}  \\
        &= \frac{\sqrt{n}}{N}\sum_{i=1}^N \frac{L_iI(T_i=k)}{\nu_n \pi_k(\bX_i)}\left\{ Y_i - \xi_k(\bV_i)\right\}  - \frac{\sqrt{n}}{N}\sum_{i=1}^N \frac{L_iI(T_i=k)}{\nu_n \pi_k(\bX_i)}\left\{ \xi(\bV_i;\bgamhat) - \xi_k(\bV_i)\right\}  + O_p(\nu_n^{1/2}),
    \end{eq*}
    where:
    \begin{eq*}
        &\frac{\sqrt{n}}{N}\sum_{i=1}^N \frac{L_iI(T_i=k)}{\nu_n} \left\{ \frac{1}{\pi_k(\bX_i;\balphhat)}- \frac{1}{\pi_k(\bX_i)}\right\}\left\{ Y_i - \xi(\bV_i;\bgamhat)\right\} \\
        &\qquad =\frac{\sqrt{n}}{N}\sum_{i=1}^N \frac{L_iI(T_i=k)}{\nu_n} \frac{\partial}{\partial\balph\trans} \frac{1}{\pi_k(\bX_i;\balphbar)}\left( \balphhat - \balphbar\right)\left\{ Y_i - \xi(\bV_i;\bgambar)\right\} \\
        &\qquad\qquad + O_p\left\{n^{1/2}\left(\norm{\balphhat-\balphbar}\norm{\bgamhat-\bgambar}+\norm{\balphhat-\balphbar}^2\right)\right\} \\
        &\qquad =O_p\left\{ n^{1/2}\left(\norm{\balphhat-\balphbar}+\norm{\balphhat-\balphbar}\norm{\bgamhat-\bgambar}+\norm{\balphhat-\balphbar}^2\right)\right\}.
    \end{eq*}
    The second term is:
    \begin{eq*}
        \Tscr_{2,k} &= -\frac{\sqrt{n}}{N} \sum_{i=1}^N \left\{\frac{I(T_i =k )}{\pi_k(\bX_i)}-1\right\}\mu_k(\bX_i;\bbetabar) + \left\{\frac{I(T_i =k )}{\pi_k(\bX_i)}-1\right\}\left\{\mu_k(\bX_i;\bbetahat) - \mu_k(\bX_i;\bbetabar)\right\} \\
        &\qquad  -\frac{\sqrt{n}}{N} \sum_{i=1}^N \left\{\frac{1}{\pi_k(\bX_i;\balphhat)}-\frac{1}{\pi_k(\bX_i)}\right\}I(T_i =k )\mu_k(\bX_i;\bbetahat) \\
        &= O_p(\nu_n^{1/2}) + O_p(N^{-1/2})O_p(1) + O_p(\nu_n^{1/2}),
    \end{eq*}
    where we use that $\E\left[\left\{ \frac{I(T_i=k)}{\pi_k(\bX_i)}-1\right\}g(\bX_i)\right]=0$ for any integrable function $g$ of $\bX_i$ and:
    \begin{eq*}
        &\frac{\sqrt{n}}{N}\sum_{i=1}^N \left\{ \frac{1}{\pi_k(\bX_i;\balphhat)}- \frac{1}{\pi_k(\bX_i)}\right\}I(T_i=k)\mu_k(\bX_i;\bbetahat) \\
        &\qquad =\frac{\sqrt{n}}{N}\sum_{i=1}^N \frac{\partial}{\partial\balph\trans} \frac{1}{\pi_k(\bX_i;\balphbar)}\left( \balphhat - \balphbar\right)\mu_k(\bX_i;\bbetabar) + O_p\left\{n^{1/2}\left(\norm{\balphhat-\balphbar}\norm{\bbetahat-\bbetabar}+\norm{\balphhat-\balphbar}^2\right)\right\} \\
        &\qquad =O_p\left\{ n^{1/2}\left(\norm{\balphhat-\balphbar}+\norm{\balphhat-\balphbar}\norm{\bbetahat-\bbetabar}+\norm{\balphhat-\balphbar}^2\right)\right\}.
    \end{eq*}
    The third term is:
    \begin{eq*}
        \Tscr_{3,k} &= \frac{\sqrt{n}}{N}\sum_{i=1}^N \left\{\frac{I(T_i = k)}{\pi_k(\bX_i)}\xi_k(\bV_i) - \mubarstr_k\right\} + 
        \frac{\sqrt{n}}{N}\sum_{i=1}^N \frac{I(T_i = k)}{\pi_k(\bX_i)}\left\{\xi_k(\bV_i;\bgamhat)-\xi_k(\bV_i) \right\} \\
        &\qquad + \frac{\sqrt{n}}{N}\sum_{i=1}^N \left\{\frac{1}{\pi_k(\bX_i;\balphhat)}-\frac{1}{\pi_k(\bX_i)}\right\}I(T_i = k)\xi_k(\bV_i;\bgamhat) \\
        &= \frac{\sqrt{n}}{N}\sum_{i=1}^N \frac{I(T_i = k)}{\pi_k(\bX_i)}\left\{\xi_k(\bV_i;\bgamhat)-\xi_k(\bV_i) \right\} + O_p(\nu_n^{1/2}) + O_p(\nu_n^{1/2}),
    \end{eq*}
    where where we use that $\E\left\{ \frac{I(T_i=k)}{\pi_k(\bX_i)}\xi_k(\bV_i)\right\} = \mubarstr_k$ and:
    \begin{eq*}
        &\frac{\sqrt{n}}{N}\sum_{i=1}^N \left\{\frac{1}{\pi_k(\bX_i;\balphhat)}-\frac{1}{\pi_k(\bX_i)}\right\}I(T_i = k)\xi_k(\bV_i;\bgamhat)\\
        &\qquad =\frac{\sqrt{n}}{N}\sum_{i=1}^N \frac{\partial}{\partial\balph\trans}\frac{1}{\pi_k(\bX_i;\balphbar)}\left(\balphhat-\balphbar\right)I(T_i = k)\xi_k(\bV_i) + O_p\left\{ n^{1/2}\left(\norm{\balphhat-\balphbar}\norm{\bgamhat-\bgambar}+\norm{\balphhat-\balphbar}^2\right)\right\} \\
        &\qquad = O_p\left\{n^{1/2}\left(\norm{\balphhat-\balphbar} +\norm{\balphhat-\balphbar}\norm{\bbetahat-\bbetabar} +\norm{\balphhat-\balphbar}^2\right)\right\}.
    \end{eq*}
    Combining the terms, we obtain:
    \begin{eq*}
        \sqrtn(\muhatkAIPW-\mubarstr_k) &= \frac{\sqrt{n}}{N}\sum_{i=1}^N \frac{L_iI(T_i=k)}{\nu_n \pi_k(\bX_i)}\left\{ Y_i - \xi_k(\bV_i)\right\}  \\
        &\qquad - \frac{\sqrt{n}}{N}\sum_{i=1}^N \frac{I(T_i=k)}{\pi_k(\bX_i)} \left( \frac{L_i}{\nu_n}-1\right) \left\{ \xi(\bV_i;\bgamhat) - \xi_k(\bV_i)\right\}  + O_p(\nu_n^{1/2}).
    \end{eq*}
    The second term can be further expanded as:
    \begin{eq*}
        &- \frac{\sqrt{n}}{N}\sum_{i=1}^N \frac{I(T_i=k)}{\pi_k(\bX_i)} \left( \frac{L_i}{\nu_n}-1\right) \left\{ \xi(\bV_i;\bgamhat) - \xi_k(\bV_i)\right\} \\
        &= - \frac{\sqrt{n}}{N}\sum_{i=1}^N \frac{I(T_i=k)}{\pi_k(\bX_i)} \left( \frac{L_i}{\nu_n}-1\right) \frac{\partial}{\partial\bgam\trans}\xi(\bV_i;\bgambar) (\bgamhat-\bgambar) + O_p(\norm{\bgamhat-\bgambar}^2).
    \end{eq*}
    We now apply a weak law of large numbers for triangular arrays (Theorem 2.2.6 of \cite{durrett2019probability}) to control the average in the first term.
    Let $\bgam=(\gamma_1,\ldots,\gamma_q)\trans$. For $j=1,2,\ldots,q$, the mean of the sum in the first term satisfies:
    \begin{eq*}
        \E\left\{ \sum_{i=1}^N \frac{I(T_i=k)}{\pi_k(\bX_i)}\left(\frac{L_i}{\nu_n}-1\right)\frac{\partial}{\partial\gamma_j}\xi(\bV_i;\bgambar)\right\} 
        &= \sum_{i=1}^N \E\left(\frac{L_i}{\nu_n}-1\right)\E\left\{ \frac{I(T_i=k)}{\pi_k(\bX_i)}\frac{\partial}{\partial\gamma_j}\xi(\bV_i;\bgambar)\right\} \\
        &= 0,
    \end{eq*}
    where the first equality follows from $L \indep \bZ$, and the variance satisfies:
    \begin{eq*}
        &N^{-2}Var\left\{ \sum_{i=1}^N \frac{I(T_i=k)}{\pi_k(\bX_i)}\left(\frac{L_i}{\nu_n}-1\right)\frac{\partial}{\partial\gamma_j}\xi(\bV_i;\bgambar)\right\} = N^{-1}\E\left\{ \frac{I(T_i=k)}{\pi_k(\bX_i)}\left(\frac{L_i}{\nu_n}-1\right)\frac{\partial}{\partial\gamma_j}\xi(\bV_i;\bgambar)\right\}^2 \\
        &\qquad\qquad= N^{-1}\left(\frac{1}{\nu_n}-1\right) \E\left\{ \frac{I(T_i=k)}{\pi_k(\bX_i)}\frac{\partial}{\partial\gamma_j}\xi(\bV_i;\bgambar)\right\}^2 \\
        &\qquad\qquad \to 0.
    \end{eq*}
By the weak law of large numbers for triangular arrays we then obtain:
    \begin{eq*}
        N^{-1}\sum_{i=1}^N \frac{I(T_i = k)}{\pi_k(\bX_i)}\left( \frac{L_i}{\nu_n}-1\right) \frac{\partial}{\partial\gamma_j\trans}\xi(\bV_i;\bgambar) = o_p(1),
    \end{eq*}
    for $j=1,2,\ldots,q$.  Collecting and simplifying the results from above, we have:
\begin{eq*}
    \sqrtn(\muhatkAIPW-\mubarstr_k) &= \frac{1}{\sqrt{n}}\sum_{i=1}^n \frac{I(T_i=k)}{ \pi_k(\bX_i)}\left\{ Y_i - \xi_k(\bV_i)\right\} + O_p(\nu_n^{1/2} + n^{-1}).
\end{eq*}
\end{proof}

\section*{Web Appendix D: Sensitivity to Tuning Parameter in Ridge Regression}
To better understand the impact of the Ridge tuning parameter in finite samples, we repeated simulations from the main 
text using a range of tuning parameters based on undersmoothing the cross-validated tuning parameter and also based on
explicitly defined tuning parameters that satisfy the $\lambda_n = o(n^{-1/2})$ condition.  

The simulations were run in the correctly specified scenario from the main text with medium strength surrogates and sample 
sizes of $n=100,500$ and $N=1000$.  The differences in root mean square error (RMSE) is minimal for different choices of 
the tuning parameter, especially for in samples with larger $n$.  The results suggest that the performance of the final 
estimators are not especially sensitive to the tuning parameter choice.\\

\bgroup
\def\arraystretch{.575}%
    \begin{tabular}{lcccc}
      \hline
      & \textbf{$n$} & \textbf{Bias} & \textbf{SE} & \textbf{RMSE} \\ 
     \hline
        $\lambda_{cv}$ & 100 & -0.015 & 0.057 & 0.059 \\ 
        $\lambda_{cv}/n^{1/10}$ & 100 & -0.012 & 0.060 & 0.061 \\ 
        $\lambda_{cv}/n^{1/2}$ & 100 & -0.010 & 0.061 & 0.062 \\ 
        $\lambda_{cv}/n^{3/4}$ & 100 & -0.010 & 0.061 & 0.062 \\ 
        $\lambda = log(p)/n^{3/5}$ & 100 & -0.016 & 0.056 & 0.058 \\ 
        $\lambda = log(p)/n^{3/4}$ & 100 & -0.010 & 0.061 & 0.062 \\ 
       \hline
       $\lambda_{cv}$ & 500 & -0.003 & 0.034 & 0.034 \\ 
       $\lambda_{cv}/n^{1/10}$ & 500 & -0.003 & 0.034 & 0.034 \\ 
       $\lambda_{cv}/n^{1/2}$ & 500 & -0.003 & 0.034 & 0.034 \\ 
       $\lambda_{cv}/n^{3/4}$ & 500 & -0.003 & 0.034 & 0.034 \\ 
         $\lambda = log(p)/n^{3/5}$ & 500 & -0.009 & 0.032 & 0.033 \\ 
         $\lambda = log(p)/n^{3/4}$ & 500 & -0.003 & 0.034 & 0.034 \\ 
          \hline
    \end{tabular}
\egroup

\end{document}